\documentclass[5p]{elsarticle}

\usepackage{amssymb,amsmath,amsfonts,amsthm}

\newtheorem{thm}{Theorem}

\newtheorem{prop}[thm]{Proposition}

\newcommand\RR{\mathbb R}

\newcommand\PP{\mathbb P}
\newcommand\Pc{\mathcal P}
\newcommand\Ac{\mathcal A}

\def\Lc{{\mathcal{L}}}
\def\Nc{{\mathcal{N}}}

\def\tg{{\mathbf{t}}}

\usepackage[usenames,dvipsnames]{xcolor}
\newcommand{\com}[1]{{\color{black} #1}}
\newcommand{\lb}[1]{{\color{black} #1}}

\begin{document}

\begin{frontmatter}

\title{Extraction of cylinders and cones from minimal point sets}

\author[1]{Laurent Bus\'e}
\author[1,2]{Andr\'e Galligo}
\author[1,3]{Jiajun Zhang}
\address[1]{INRIA Sophia Antipolis - M\'editerann\'ee}
\address[2]{Laboratoire J.-A. Dieudonn\'e, Universit\'e de Nice Sophia Antipolis}
\address[3]{\'Ecole Polytechnique de l'Universit\'e de Nice Sophia Antipolis}
\begin{abstract}
We propose new algebraic methods for extracting cylinders and cones from minimal point sets, including oriented points. More precisely, we are interested in computing efficiently cylinders through a set of three points, one of them being oriented, or through a set of five simple points. We are also interested in computing efficiently cones through a set of two oriented points, through a set of four points, one of them being oriented,  or through a set of six points. For these different interpolation problems, we give optimal bounds on the number of solutions. Moreover, we describe algebraic methods targeted \com{to solve these problems efficiently}.
\end{abstract}


\end{frontmatter}

\section{Introduction}

 \com{Extracting} geometric primitives from 3D point clouds is an important problem in reverse engineering. These 3D point clouds are typically obtained \com{from} accurate 3D scanners and there exist several methods for \com{extracting} 3D geometric primitives \cite{PA}. An important category among these methods \com{is based} on the RANSAC approach \cite{Fischler:1981:RSC:358669.358692,schnabel-2007-efficient,PA}. For such methods, the primitives are \com{extracted directly}  from the input point cloud. The basic idea is to extract a particular elementary type of shape, such as planes, spheres, cylinders, cones or tori, from the smallest possible set of points and then to judge if this extracted primitive is relevant to the full point cloud. Therefore, for this category of methods it is very important to compute a particular type of shape through the smallest possible number of points, including normals if available. If \com{extracting} planes and spheres is easy, the cases of cylinders and cones are more involved. In this paper we provide new methods for extracting these geometric primitives from the smallest possible number of points, counting multiplicities of oriented points (i.e.~points given with their normal vector). \com{These methods are intended to serve the larger goal of improving speed and numerical accuracy in data extraction from graphical information.}   As far as we know, and surprisingly, the above-mentioned problems have not appeared in the existing literature with the exception of \cite{Devillers02,Lichtblau12}. Instead, the classical approaches to these interpolation problems usually extract, actually we should say estimate, these geometric primitives from an overdetermined number of points, counting multiplicities (e.g.~\cite{Lukacs:1998aa}).

\medskip

An \emph{oriented point} is a couple of a point and a nonzero vector. A surface is said to interpolate an oriented point if the point belongs to the surface and its associated vector is \com{collinear} to the normal of the surface at this point, we do not assume that the orientation of the normal of the point is the same as the orientation of the surface since often in the data sets normals are unoriented. Moreover, it is important to deal with inhomogeneous data, that is to say some points are oriented but not all, in order to take into account the estimated accuracy of oriented point clouds that are \com{generated} by means of normal estimation algorithms. \com{Data} made of points and oriented points will be called a \emph{mixed set of points}.

\medskip 

We emphasize that interpolating at a point imposes a single algebraic condition on a given shape whe\-re\-as interpolating at an oriented 3D point imposes three algebraic conditions. Typically, a 3D plane is uniquely defined either by three distinct points or by one oriented point. A sphere is uniquely defined either by four points or by one oriented point and an additional point. In these two cases, it turns out that there is a unique shape that interpolates a mixed set of points corresponding to the number of parameters of this shape (a plane is determined by three parameters and the sphere is determined by four parameters). In this paper, we will treat interpolation of two other basic shapes, namely cylinders and cones for which the situation is more involved.

\medskip

Our approach is inspired by effective methods in algebraic geometry.
\lb{We consider two families of unknowns. The first one corresponds to the parameters needed to describe all features of the targeted
surface (e.g. the radius and axis of a circular cylinder) and hence its equation. The second family consists of auxiliary
unknowns (e.g. such as a special point on that axis)  which permit us to describe a collection of
geometric constructions.} These constructions are designed to establish a complete link between the input and
the first family of unknowns.
Then, we translate algebraically the collection of con\-str\-aints attached to theses geometric constructions into a system of polynomial equations
that we further analyze and simplify, discarding spurious solutions if \com{necessary}.
Since the input and output are (and should be) real approximate data, we designed efficient algorithms to compute
very accurate real solutions of these systems of equations.
Indeed, in all the considered cases, we were able to express the results as the solutions of (generalized) \com{eigenvalue problems}
together with close formulas. \com{These expressions allow us} to rely on classical matrix \com{computation} software and achieve accuracy and efficiency.
Prototypes of our algorithms \com{are} implemented in the computer algebra system {\tt MAPLE}, and we provide some statistics and
timings (which are quite satisfactory).

\section{Interpolation of cylinders}

A cylinder (more precisely a right circular cylinder) is defined as the set of points in the three-dimensional affine space $\RR^3$ located at a fixed distance (called the \emph{radius} of the cylinder) of a given straight line (called the \emph{axis} of the cylinder). It is hence defined by means of five parameters~: four parameters describe a line in $\RR^3$ and an additional parameter measures the radius.

\medskip

A popular determination of a cylinder is done \com{by interpolating} two points with normals, which imposes six conditions (instead of five). So, a priori no cylinder interpolates \com{this data}; therefore some approximations are necessary.  In this section, we will give new methods to compute  \com{cylinders} using just five independent conditions. There are two possible such minimal configurations, either an oriented point and two other distinct points, or five distinct points.

\subsection{Cylinders through a mixed minimal point set}

\com{We seek} the cylinders that interpolate a given mixed minimal set of points $\Pc$. Since a cylinder is given by 5 parameters, $\Pc$ is assumed to be composed of an oriented point $p_1$ with its normal vector $n_1$ and two other distinct points $p_2, p_3$ in $\RR^3$. 

\medskip

First, by a linear change of coordinates, one can assume that $p_1=(0,0,0)$ and  $n_1=(0,0,1)$ and we set $p_2=(x_2,y_2,z_2)$ and $p_3=(x_3,y_3,z_3)$. Then, the axis of a cylinder interpolating $\Pc$ must be orthogonal to the $z$-axis and must intersect it. It follows that a normal plane $\Pi$ contains the $z$-axis, hence is given by an equation of the form $lx+my=0$ where $t:=(l,m,0)$ is the corresponding direction of the axis. Observe that these directions are in correspondence with a projective line $\PP^1$. For simplicity, we set $\rho:=\sqrt{l^2+m^2}=\|t\|>0$.

Now, we compute the orthogonal projections $q_1$ and $q_2$ of $p_2$ and $p_3$ onto the plane $\Pi$. $\Pi$ contains the point $p_1$ and is generated by the two orthogonal vectors $n_1$ and $v=n_1\wedge t=(-m,l,0)$. The matrix
$$M=
\left(
\begin{array}{ccc}
	\frac{l}{\rho} & \frac{m}{\rho} & 0 \\
	- \frac{m}{\rho} & \frac{l}{\rho} & 0 \\
	0 & 0 & 1
\end{array}
\right)$$
defines the change of coordinates from the current coordinate system $(x,y,z)$ to a new coordinate  system $(x',y',z')$, with the same origin $p_1$, defined by the three vectors $t/\rho,$ $v/\rho,$ $n_1$, where $\Pi$ has equation $x'=0$. It follows that the coordinates $(x_i',y_i',z_i')$ of $p_i$, $i=2,3$, in this new coordinate system are given by
$$(x_i',y_i',z_i')=\left(x_i\frac{l}{\rho}+y_i \frac{m}{\rho}, - x_i\frac{m}{\rho}+ y_i\frac{l}{\rho} , z_i\right).$$
Therefore, the coordinates of the orthogonal projections $q_i$, $i=2,3$ are given by
$$q_i=\left(- x_i\frac{m}{\rho}+ y_i\frac{l}{\rho} , z_i\right) \in \Pi$$
in the basis $v/\rho,n_1$. 

The existence of a cylinder interpolating $\Pc$ is equivalent to the fact that the points $p_1$, $q_2$ and $q_3$ \com{all belong to} a circle whose center $c$ is located on the $z$-axis, say $c=(0,0,r)$. Such a circle has an equation of the form $y'^2+(z'-r)^2=r^2$, or equivalently $y'^2+z'^2-2rz'=0$. Therefore, this cocyclicity condition can be written as
\begin{multline*}
0=\left| 
\begin{array}{cc}
 y_2'^2+z_2^2 & z_2 \\
y_3'^2+z_3^2 & z_3 
\end{array}
\right| \\
=\frac{1}{\rho^2}
\left| 
\begin{array}{cc}
\left(- x_2m+ y_2l\right)^2 + (l^2+m^2)z_2^2 & z_2 \\
\left(- x_3m+ y_3l\right)^2 + (l^2+m^2)z_3^2 & z_3 
\end{array}
\right|.
\end{multline*}
Since $\rho>0$, the expansion of this latter determinant \com{allows us to} rewrite this condition as a
degree 2 homogeneous equation $al^2+blm+cm^2$ where the coefficients $a,b,c$ are given by the following closed formulas
$$
a:=\left|
\begin{array}{cc}
	y_2^2+z_2^2 & z_2 \\
	y_3^2+z_3^2 & z_3
\end{array}\right|, \ \
b:=-2\left|
\begin{array}{cc}
	x_2y_2 & z_2 \\
	x_3y_3 & z_3
\end{array}\right|, $$
$$c:=\left|
\begin{array}{cc}
	x_2^2+z_2^2 & z_2 \\
	x_3^2+z_3^2 & z_3
\end{array}\right|.
$$
   
Unless $a=b=c=0$, this equation has two roots, counting multiplicities, in the field of complex numbers.  If a real solution is found, that is to say the direction of a real cylinder interpolating $\Pc$ (observe that one can impose $\rho=1$ since the condition is homogeneous in $l,m$), then the remaining parameter $r$ is uniquely determined by  
one of the formulas
\begin{align}\label{eq:z0}
	2z_2r&=y_2'^2 + z_2^2=\frac{1}{\rho^2}\left(- x_2m+ y_2l \right)^2+z_2^2, \\
	 2z_3r&=y_3'^2+z_3^2=\frac{1}{\rho^2}\left(- x_3m+ y_3l \right)^2 + z_3^2,	
\end{align}
depending \com{on} whether $z_2\neq 0$ or $z_3\neq 0$. We notice that if $z_2=z_3=0$ then $a=b=c=0$. 

\begin{thm}\label{thm:cyl1N2P} Given a mixed set of points $\Pc$ composed of an oriented point $p_1, n_1$ and two other points $p_2,p_3$ such that 
	\begin{itemize}
		\item[i)] $p_1,p_2,p_3$ are all distinct,
		\item[ii)] $p_1,p_2,p_3$ do not belong to \com{a common plane that} is normal to $n_1$,
		\item[iii)] $p_2$ and $p_3$ are not symmetric with respect to the line through $p_1$ and generated by $n_1$, 
	\end{itemize} 
then there exist at most 2 real cylinders interpolating $\Pc$. Otherwise, \com{there exists} a cylinder (possibly "flat", i.e.~with infinite radius) interpolating $\Pc$ in any direction \com{that} is normal to $n_1$.
\end{thm}
\begin{proof}
Following the above discussion, this theorem will be proved if we show that $a=b=c=0$ if and only if at least one of the three conditions i), ii), iii) holds. It is not hard to check that if one of the three latter conditions holds then $a=b=c=0$. To prove the converse, we observe that by a linear change of coordinate, we can assume that $x_2=0$ in addition to the fact that $p_1=(0,0,0)$ and $n_1=(0,0,1)$. Then, we have $b=2x_3y_3z_2$ and hence three cases to analyze.

If $x_3=0$, then $c=z_2z_3(z_2-z_3)$ so that $c=0$ implies that $z_2=0$, or $z_3=0$ or $z_2=z_3$. If $x_3=z_2=0$, then $a=y_2^2z_3$ so that $a=0$ implies that i) or ii) hold. Similarly, if $x_3=z_3=0$ then $a=y_3^2z_2$ so that $a=0$ implies that i) or ii) hold. Finally, if $x_3=0$ and $z_2=z_3$ then 
$a=z_2(y_2^2-y_3^2)$ so that $a=0$ implies that i) or iii) hold.  

The case corresponding to $y_3=0$ can be treated exactly as the previous case $x_3=0$,  exchanging $x_2$ with $x_3$ and $y_2$ with $y_3$ \com{leaves $b$ unchanged}  and permute $a$ and $c$.

Finally, if $z_2=0$, then $c=0$ (recall $x_2=0$) and $a=y_2^2z_3$. So $a=0$ if either $y_2=0$ or $z_3=0$. But $y_2=0$ means that $p_1=p_2$, i.e. i) holds, and $z_3=0$ means that ii) holds. 	
\end{proof}

When a mixed set of points $\Pc$ satisfies conditions i), ii) and iii) in \com{Theorem \ref{thm:cyl1N2P}}, then $a$, $b$ and $c$ are not all zero and hence \com{whether there are} zero, one or two homogeneous real solutions to the equation $al^2+blm+cm^2=0$ is decided by means of the discriminant  $\Delta:=b^2-4ac$ which depends on the coordinates of $p_2$ and $p_3$. If $\Delta<0$ then there is no real homogeneous solution, if $\Delta=0$ then \com{there exists} a double homogeneous solution and if $\Delta>0$ then there exists two distinct homogeneous solutions. As we have already observed in the proof of Theorem \ref{thm:cyl1N2P}, it is possible to assume, without loss of generality, that $x_2=0$ in addition of $p_1=(0,0,0)$ and $n_1=(1,0,0)$. Then, a straightforward computation shows that 
\begin{multline*}
\Delta=
4 {  z_2} {  z_3}  \left( {{  y_2}}^{2} \left( {{  x_3}}^{2}+{{
  z_3}}^{2} \right) +{{  z_2}}^{2} \left( {{  x_3}}^{2}+{{  y_3}}^
{2} \right)\right. \\ \left. -{  z_2} {  z_3}  \left( {{  x_3}}^{2}+{{  y_2}}^{2}
+{{  y_3}}^{2}+ \left( {  z_2}-{  z_3} \right) ^{2} \right) 
 \right).	
\end{multline*}
\com{From} this equation we see directly that there \com{are} no real cylinders interpolating $\Pc$ if $p_2$ and $p_3$ \com{are not} on the same side of the plane through $p_1$ and normal to $n_1$ (i.e.~$z_2$ and $z_3$ have opposite signs). Another interesting case is to assume that $p_2$ belongs to the plane through $p_1$ and normal to $n_1$ (i.e.~$z_2=0$). Indeed, in this case $\Delta=0$ so there is a unique cylinder (counted with multiplicity two) through $\Pc$ : its direction is given by  $t=(-b,2a,0)=(2x_2y_2z_3,2y_2^2z_3,0)$ and $r$ is still defined by \eqref{eq:z0} ($z_3$ is assumed to be nonzero for otherwise $a=b=c=0$).

Finally, we notice that if $a=0$ then the directions are given by the equation $blm+cm^2=0.$
It follows that these directions are given by $(1,0,0)$ and $(-c,b,0)$. Of course, if in addition $b=0$ then there is a single \com{interpolating} cylinder that appears with multiplicity 2 (the discriminant vanishes). The radius \com{$r_1$ or $r_2$ is} then deduced from \eqref{eq:z0}. Similarly, if $c=0$ then the directions of the interpolating cylinders are given by $(0,1,0)$ and $(-b,a,0)$.

\paragraph{Extraction algorithm} From the above analysis, we deduce the following extraction algorithm \com{for cylinders} from a mixed point set $\Pc$:
\begin{enumerate}
	\item Perform a linear change of coordinates so that the oriented point is at the origin and its normal vector is $(0,0,1)$. If $z_2z_3<0$ then there is no solution.
	\item Compute the quantities $a$, $b$ and $c$. If $a=b=c=0$ then stop because there are infinitely many interpolating cylinders. 
	\item If $a=0$ or $c=0$ then compute the interpolating cylinders as explained above and stop.
	\item Otherwise, solve the equation $al^2+blm+cm^2=0$ and keep only those roots that are real numbers (up to numerical precision). For each such root, compute the corresponding radius by means of \eqref{eq:z0} and return the interpolating cylinder. 
\end{enumerate} 

This algorithm has been implemented with the {\tt Maple} software and all timings are measured on a Mac laptop equipped with Intel Core i7 CPU @ 2.8GHz, 16 GB memory. We observed that \com{computing} the cylinders through a random set of points takes \com{on average} 3.5ms \com{(including all the steps \lb{in the above algorithm})} and is almost constant, \com{i.e.~independent} of the point set. 
The proportion of the number of cylinders found through random point sets is given in Table \ref{tab:cylinders} and an illustrative picture is given in Figure \ref{fig:cyl1N2P}.

\begin{figure}[ht!]
	\centering
	\includegraphics[width=60mm]{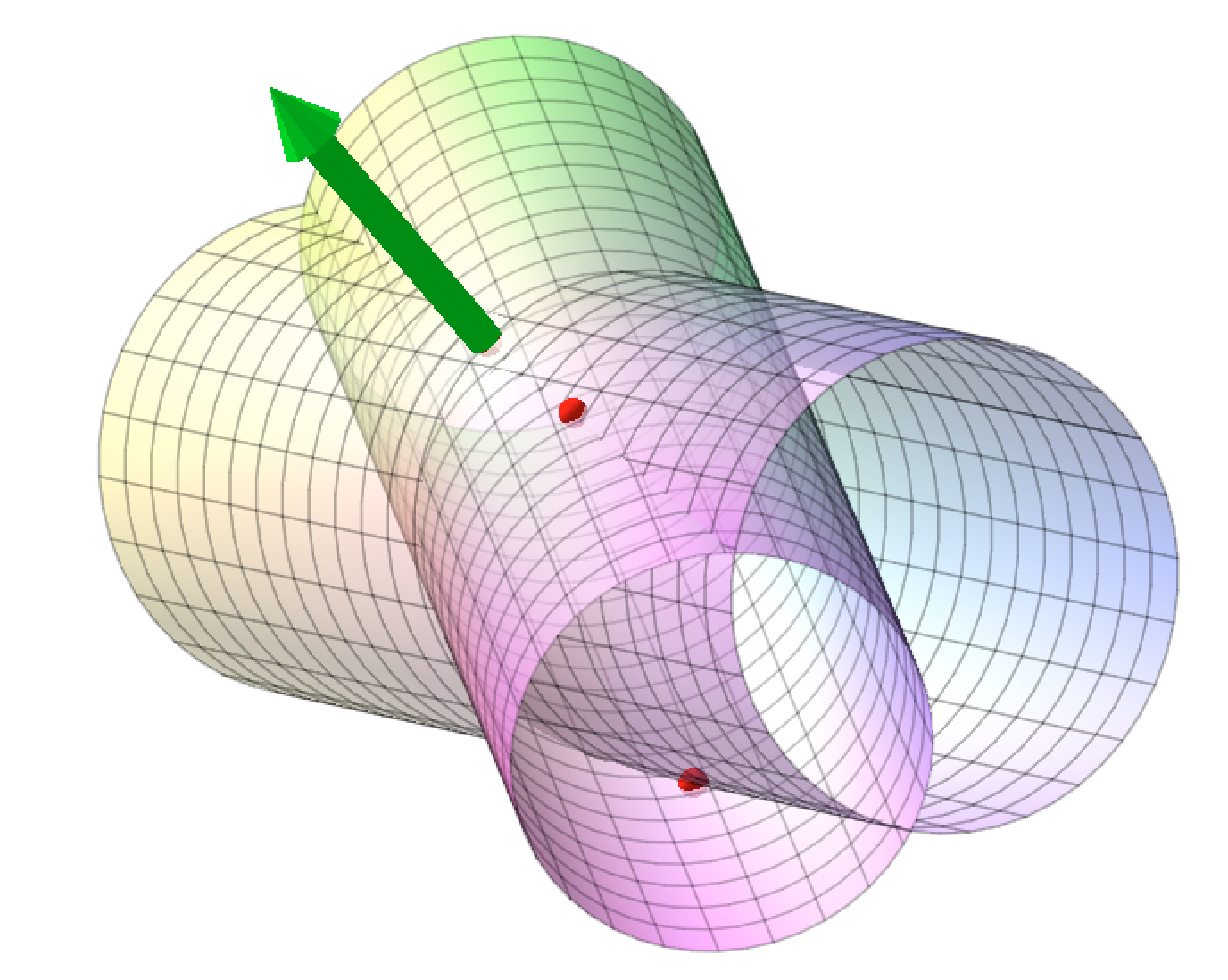}		
\caption{A general mixed set of points with the two interpolated cylinders.}\label{fig:cyl1N2P}	
\end{figure}

\begin{table}[ht!]
\begin{center}
\begin{tabular}{c|c|c|}
Number of cylinders	 & 0 & 2  \\ 
	 \hline
Proportion (\%)	 &  46.87  & 53.13 \\
\hline	
\end{tabular}
\caption{Proportion of the number of cylinders found through a thousand random point sets.}\label{tab:cylinders}
\end{center}
\end{table}

\subsection{Cylinders through five points}

The problem of extracting a cylinder passing through five points has already been treated in the literature, see \cite{Devillers02,BKM05,Lichtblau12} \lb{(see also \cite{Chaperon:2001:ECF:647260.718511,Beder:2006:DSC:2094437.2094451})}. These works solve the problem by relying on a polynomial system solver. In this section, we briefly review the model already described in \cite{Devillers02} and then we propose an improvement to gain efficiently. Our key ingredient removes some spurious solutions by means of additional algebraic manipulations. 

\subsubsection{Geometric analysis}

\lb{Given} a set $\Pc$ of five points $p_1, p_2, p_3, p_4, p_5$ \lb{we seek cylinders} through $\Pc$. First of all, \lb{by a linear} change of coordinates, we assume without loss of generality that the coordinate system $(x,y,z)$ is such that
\begin{equation}\label{eq:hyp-p1p2p3p4p5}
p_1=(0,0,0), \ p_2=(x_2,0,0), \ p_3=(x_3,y_3,0),	
\end{equation} 
$$p_4=(x_4,y_4,z_4), \ p_5=(x_5,y_5,z_5). $$

Let us pick a nonzero vector $t=(l,m,n)$ (which will represent the direction of the axis) and denote by $\Pi$ the plane through the origin which is orthogonal to $t$. Let us also denote by $q_i$ the orthogonal projection of $p_i$ onto $\Pi$, for all $i=1,\ldots,5$. Then, $p_1,\ldots,p_5$ belongs to a cylinder of direction $t$ if and only if the coplanar points $q_1,\ldots,q_5$ are cocyclic. In order to make explicit \lb{this cocyclic condition}, we consider a new system of coordinates $(x',y,',z')$ whose third axis is parallel to $t$. More precisely, we choose the coordinates system $(x',y',z')$ which is obtained from the coordinate system $(x,y,z)$ by the change of coordinates defined by the orthogonal matrix
$$M=\left( \begin{array}{ccc}
           \frac{\sqrt{m^{2}+n^{2}}}{\|\tg\|} & -\frac{lm}{\|\tg\|\sqrt{m^{2}+n^{2}}} & -\frac{ln}{\|\tg\|\sqrt{m^{2}+n^{2}}} \\
           0 & \frac{n}{\sqrt{m^{2}+n^{2}}} & -\frac{m}{\sqrt{m^{2}+n^{2}}} \\
           \frac{l}{\|\tg\|} & \frac{m}{\|\tg\|} & \frac{n}{\|\tg\|} \\
           \end{array} \right).$$
Therefore, the coordinates of $q_i$ in the system $(x',y',z')$ are given by
\begin{multline*}
\left( x_{i} \frac{\sqrt{m^{2}+n^{2}}}{\|\tg\|}  - \frac{lm}{\|\tg\|\sqrt{m^{2}+n^{2}}} y_i -\frac{ln}{\|\tg\|\sqrt{m^{2}+n^{2}}} z_i , \right. \\
 \left.        \frac{n}{\sqrt{m^{2}+n^{2}}} y_i  -\frac{m}{\sqrt{m^{2}+n^{2}}} z_i , 0 \right)=:(x_i',y_i',0).
\end{multline*}
\com{Now} the points $q_1,q_2,q_3$ and $q_4$ are cocyclic (or aligned, which corresponds to a circle of infinite radius) providing 
$$\left|
           \begin{array}{cccc}
           1 & 1 & 1 & 1 \\
           x'_1 & x'_2 & x'_3 & x'_4 \\
           y'_1 & y'_2 & y'_3 & y'_4 \\
           x_1'^2+y_1'^2 & x_2'^2+y_2'^2 & x_3'^2+y_3'^3 & x_4'^2+y_4'^2 \\
           \end{array}
  \right|=0.$$
Since $x'_1=y'_1=0$, $x'_2=\frac{\sqrt{m^{2}+n^{2}}}{\|t\|} x_2, \ y'_2=0$ and
$$x_i'^2+y_i'^2=\|q_i\|^2=\|p_i\|^2-\frac{(\tg.p_i)^2}{\|\tg\|^{2}}$$
for $i=3,4$, we deduce that the points $p_1,p_2,p_3,p_4$ all belong to a cylinder  
of direction $\tg$ if and only if 
$$C_{p_1,p_2,p_3,p_4}(l,m,n)=0$$ 
where, after some \com{calculations},
\begin{multline*}
C_{p_1,p_2,p_3,p_4}(l,m,n)  :=  
x_2^2(m^2+n^2)\left|
           \begin{array}{ccc}
             l & x_3 & x_4 \\
             m & y_3 & y_4 \\
             n & 0 & z_4 \\
             \end{array} \right|  \\
	    - x_2\left|
            \begin{array}{ccc}
             m & y_3 & y_4 \\
             n & 0 & z_4 \\
             0 & \|\tg\|^{2}\|p_3\|^2-(t.p_3)^2 & \|\tg\|^{2}\|p_4\|^2-(t.p_4)^2 \\
             \end{array}\right|.	
\end{multline*}
This condition is given by a homogeneous polynomial of degree 3 in $l,m,n$. Observe that \com{this polynomial} is satisfied for the six particular directions corresponding to the lines $(p_{i}p_{j})$, $1\leq i< j \leq 4$. To determine the cylinders through the five points 
$p_{1},\ldots,p_{5}$, we consider the polynomial system of equations
\begin{equation}\label{eq:C12345}
	C_{p_{1},p_{2},p_{3},p_{4}}(l,m,n)=C_{p_{1},p_{2},p_{3},p_{5}}(l,m,n)=0
\end{equation}
\com{Equation \ref{eq:C12345}} corresponds geometrically to the intersection of two cubic curves, which is composed of nine points whenever this intersection is finite. As already \com{noted}, the directions corresponding to the lines $(p_{1}p_{2})$, $(p_{1}p_{3})$ and $(p_{2}p_{3})$ are roots of these two equations, but \com{these roots} do not imply the existence of a cylinder interpolating $\Pc$, they have to be removed once (if there is indeed a cylinder in this direction, then it will appear as a multiple solution). Moreover, when a direction of a cylinder through $\Pc$ is found, then it is completely determined since its radius and a point on its axis are given by the center and the radius of the circle through the points $q_1,q_2,q_3,q_4,q_5$. Therefore, we have the following result.

\begin{thm}[\cite{Devillers02}] If the algebraic system of equations \eqref{eq:C12345} has a finite number of solutions, then there may be zero, two, four or six real cylinders through $\Pc$.	
\end{thm}

In \cite{Devillers02,BKM05,Lichtblau12}, various polynomial system solvers are used in order to extract the directions of the cylinders through $\Pc$. We will describe a new approach where we first simplify this algebraic system and then rely on eigen-computations.   

\com{To begin}, we observe that the three extraneous directions $(p_{1}p_{2})$, $(p_{1}p_{3})$ and $(p_{2}p_{3})$ which \com{we write as} 
$$(x_2,0,0), \ (x_3,y_3,0), \ (x_3-x_2,y_3,0),$$
are the common roots of both algebraic equations  
$$F:=nx_{2}=0, \ \ G:=(mx_3-mx_2-ly_3)(mx_3-ly_3)mx_2=0.$$
Then a result in algebraic geometry says that the two polynomials $C_{p_{1},p_{2},p_{3},p_{j}}$, $j=4,5$, belong to the ideal generated by $F$ and $G$; and suggests {explicitly computing} these membership relations. Computing the Euclidean pseudo-division of $C_{p_{1},p_{2},p_{3},p_{4}}$ and $C_{p_{1},p_{2},p_{3},p_{5}}$ by the polynomial $x_2n$, we get two homogeneous polynomials of degree 2,
 \begin{multline*}
 	D_j(l,m,n)=
  \left( -{ y_j}{{ y_3}}^{2}+{{ y_j}}^{2}{ y_3}+{{ z_j}}^{2
 }{ y_3} \right) {l}^{2} \\ 
 + \left(  \left( 2{ y_j}{ y_3}{ 
 x_3}-2{ y_j}{ x_j}{ y_3} \right) m+{ z_j} \left( { 
 x_2}{ y_3}-2{ x_j}{ y_3} \right) n \right) l+ \\
 \left( { 
 x_2}{ y_j}{ x_3}-{ x_2}{ x_j}{ y_3}-{ y_j}{{ 
 x_3}}^{2}+{{ x_j}}^{2}{ y_3}+{{ z_j}}^{2}{ y_3} \right) {m}^{2} \\
 +{ z_j} \left( -{ x_2}{ x_3}+{{ x_3}}^{2}+{{ y_3}}^{2}-2
 { y_j}{ y_3} \right) nm+ \\ 
 \left( { x_2}{ y_j}{ x_3}-{
  x_2}{ x_j}{ y_3}-{ y_j}{{ x_3}}^{2}+{{ x_j}}^{2}{
  y_3}-{ y_j}{{ y_3}}^{2}+{{ y_j}}^{2}{ y_3} \right) {n}^{2
 }
 \end{multline*}
where $j=4,5$, such that
$$
\left(
\begin{array}{cc}
D_{4} & z_{4} \\
D_{5} & z_{5} 
\end{array}
\right)
\left(
\begin{array}{c}
F \\ G
\end{array}
\right)
=
\left(
\begin{array}{c}
 C_{p_{1},p_{2},p_{3},p_{4}} \\ C_{p_{1},p_{2},p_{3},p_{5}}
\end{array}
\right).
$$

\begin{prop}\label{eq:systCD} If $z_4\neq 0$ (resp.~$z_5\neq0$), then the directions of the cylinders interpolating $\Pc$ \com{correspond to} all the common roots of both equations $C_{p_1,p_2,p_3,p_4}(l,m,n)=0$  and 
	$\Delta(l,m,n):=z_{5}D_{4}(l,m,n)-z_{4}D_{5}(l,m,n) = 0$ (resp. $C_{p_1,p_2,p_3,p_5}(l,m,n)=0$ and $\Delta(l,m,n)=0$).	
\end{prop}
\begin{proof} If $z_4\neq 0$ then the system \eqref{eq:C12345} is obviously equivalent to $C_{p_1,p_2,p_3,p_4}=0$ and
	\begin{multline*}
	z_4C_{p_1,p_2,p_3,p_5}(l,m,n)-z_5C_{p_1,p_2,p_3,p_4}(l,m,n)= \\ nx_2(z_4D_5-z_5D_4)=0.
	\end{multline*}
The solutions corresponding to $n=0$ satisfy $n=G=0$ so that they define exactly the three extraneous 	directions $(p_{1}p_{2})$, $(p_{1}p_{3})$ and $(p_{2}p_{3})$. 
Therefore, the remaining polynomial system 
\begin{equation}\label{eq:andre}
C_{p_1,p_2,p_3,p_4}(l,m,n)=0, \ \ \Delta(l,m,n)=0	
\end{equation}
gives exactly the six solutions of interest. The case $z_5\neq 0$ is treated similarly.
\end{proof}

\com{Notice} that the case where $z_4=z_5=0$ corresponds geometrically to five coplanar points. If these points are also aligned, \com{there are} infinitely many cylinders. 
\com{
If they belong to an \com{ellipse}, there are two \com{non-degenerate} cylinders \com{through} these points (they are symmetric with respect to the plane $z=0$); if this ellipse is a circle these two solutions coincide. Otherwise there are no cylinders through these points.
In all cases, \com{these points} belong to a \com{degenerate} ``flat'' cylinder (infinite radius).
}
\subsubsection{Solving via eigencomputations}\label{subsec:cyl5P}

\com{Now we describe a method that allows us to compute efficiently the roots of the system of algebraic equations given in Proposition \ref{eq:systCD} as the eigenvalues of a pencil of matrices given in closed form. This approach is based on known techniques (see e.g.~\cite{Stetter:1996:MEH:242961.242966}, \cite[\S 1]{Cox05})  \lb{that allow us} to recover the solutions in a single computation step (\lb{similar to} \cite[Appendix]{Devillers02} or \cite{BKM05}) \lb{with good control on}  numerical stability and accuracy \cite{EM95}. Hereafter, \lb{we provide, for the convenience of the reader, a short review, adapted to our context, of this matrix-based solution method}.} 

\medskip

First, in order to treat \com{separately} the case $n=0$, observe that the common roots of $C_{p_1,p_2,p_3,p_5}=0$ and $\Delta=0$ such that $n=0$ are easily computed since \com{these roots} are among the three directions $(p_{1}p_{2})$, $(p_{1}p_{3})$ and $(p_{2}p_{3})$ (i.e.~the roots of $n=G=0$) and the decision is given by the evaluation of the equation $\Delta(l,m,n)$ at these directions. For instance, the direction $(1,0,0)$ is a solution if and only if 
\begin{multline*}
\Delta(1,0,0)= z_5\left(-{ y_4}{{ y_3}}^{2}+{{ y_4}}^{2}{ y_3}+{{ z_4}}^{2
 }{ y_3} \right) \\ - z_4\left( -{ y_5}{{ y_3}}^{2}+{{ y_5}}^{2}{ y_3}+{{ z_5}}^{2
 }{ y_3} \right)=0.	
\end{multline*}
Therefore, from now on we set $n=1$. 

We consider the Sylvester Matrix $S$ of the polynomials $C_{p_1,p_2,p_3,p_5}(l,m,1)$ and $\Delta(l,m,1)$ seen as univariate polynomials in the variable $m$. This matrix is a polynomial matrix of degree 2 in $l$. More precisely, \com{this matrix} is of the form $M_2l^2+M_1l+M_0$ where each $M_i$ is a $5\times 5$-matrix whose coefficients are given in closed forms in terms of $\Pc$. For instance, 
$$
M_2:=\left( 
\begin {array}{ccccc} 0&0&{  x_2}\,{{  y_3}}^{2}{  
z_4}& b &0
\\ \noalign{\medskip}0&0&0&{  x_2}\,{{  y_3}}^{2}{  z_4}&b 
\\ \noalign{\medskip}0&0& a &0&0
\\ \noalign{\medskip}0&0&0&a&0
\\ \noalign{\medskip}0&0&0&0&a
\end {array} \right) 
$$
where 
\begin{multline*}
a:=-{{  y_3}}^{2}{  y_4}\,{  z_5}+{{
  y_3}}^{2}{  z_4}\,{  y_5}+{  y_3}\,{{  y_4}}^{2}{  
z_5}-{  y_3}\,{  z_4}\,{{  y_5}}^{2} \\ +{  y_3}\,{{  z_4}
}^{2}{  z_5} -{  y_3}\,{  z_4}\,{{  z_5}}^{2}
\end{multline*}
and
$$b:=-{  x_2}\,{{  y_3}}^{2}{  y_4}+{  x_2}\,{  y_3}\,
{{  y_4}}^{2}+{  x_2}\,{  y_3}\,{{  z_4}}^{2}.$$	
Now, following \cite{BKM05}, we linearize this polynomial matrix by considering its companion matrices that are defined by 
$$
 A= \left( \begin{array}{cc}
0 & I  \\ M_0^t & M_1^t \end{array}\right), \ \ 
B= \left( \begin{array}{cc}
I & 0 \\ 0 & -M_{2}^t  \end{array}\right)
$$
where $I$ stands for the $5\times 5$-identity matrix and $M^t_i$ stands for the transpose of the matrix $M_i$. These matrices are of size $10\times 10$ and their main feature and importance is the fact that for all $\lambda \in \mathbb{C}$ and all vector $v\in \mathbb{C}^m$ we have
$$S^t(\lambda)v=0\Leftrightarrow (A-\lambda B) \left( \begin{array}{c} v\\\lambda v \end{array}
\right)=0.$$ 
In other words, the solutions to the system 
$$C_{p_1,p_2,p_3,p_4}(l,m,1)=0, \ \ \Delta(l,m,1)=0$$ 
can be computed from the eigenvalues and eigenvectors of the pencil $A,B$ which is given in closed form in terms of the input data $\Pc$. We refer the reader to \cite{BKM05} for more details on these computations. 
 
\subsubsection{Extraction algorithm} 

From the above analysis, we deduce the following extraction algorithm \com{for cylinders} through a point set $\Pc$:
\begin{enumerate}
	\item Perform a linear change of coordinates so that the five points are of the form \eqref{eq:hyp-p1p2p3p4p5}. 
	\item If $z_4=z_5=0$ the points are coplanar and the algorithm stops here.  
	\item Assume $z_4\neq 0$. \com{Then instantiate} the pre-computed matrices $M_0, M_1, M_2$, and hence  $A,B$, with the coordinates of the points $p_i$, $i=1,\ldots,5$. 
	\item Compute the six finite eigenvalues of the pencil $A,B$ (see \cite{BKM05} for details), and sort them in order to keep only those \com{that} are real numbers (up to a given precision).  
	\item For each real eigenvalue obtained \com{in Step 4}, compute the remaining coordinate of the direction by means of the associated eigenvectors (see \cite{BKM05} for details); keep those directions that are given by real numbers.
	\item For each real direction of a cylinder through $\Pc$, compute $q_1,q_2,q_3$ then the radius and the center of their circumcircle.
\end{enumerate}

This algorithm has been implemented with the {\tt Maple} software. We observed that the computation of the cylinders through a random set of points takes \com{on average} 15ms \com{(including all the steps in the above algorithm)} and is almost constant,   \com{i.e.~independent} of the point set. The proportion of the number of cylinders found through random point sets is given in Table \ref{tab:cylinders5P}; \com{we notice that we recover essentially those that appear in \cite[\S 4.1]{Lichtblau12}}. Some illustrative configurations are presented in Figure \ref{fig:cyl5P}.

\begin{figure}[ht!]
\centering	
   \includegraphics[width=50mm]{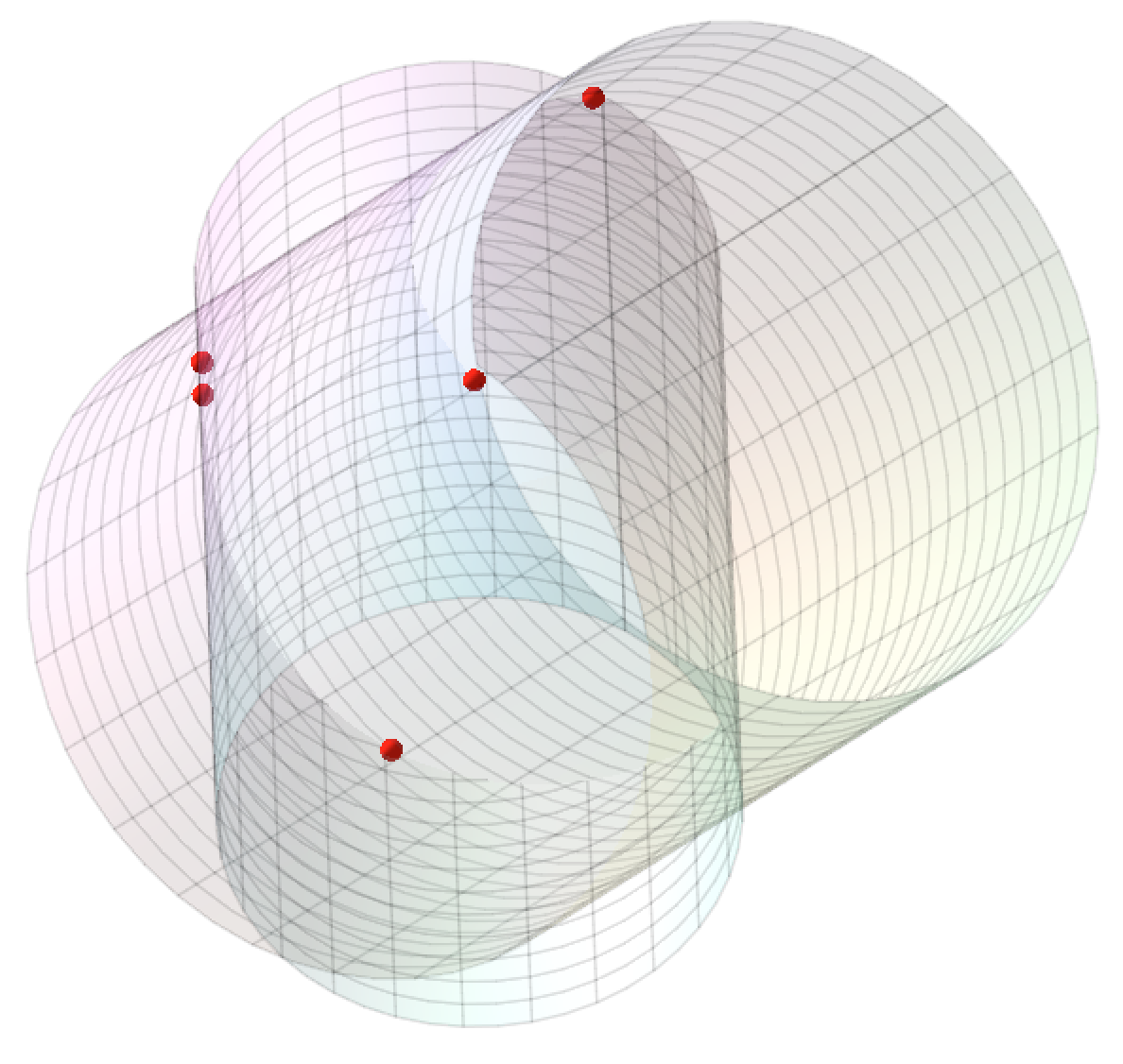}
   \hspace{4em}
   \includegraphics[width=60mm]{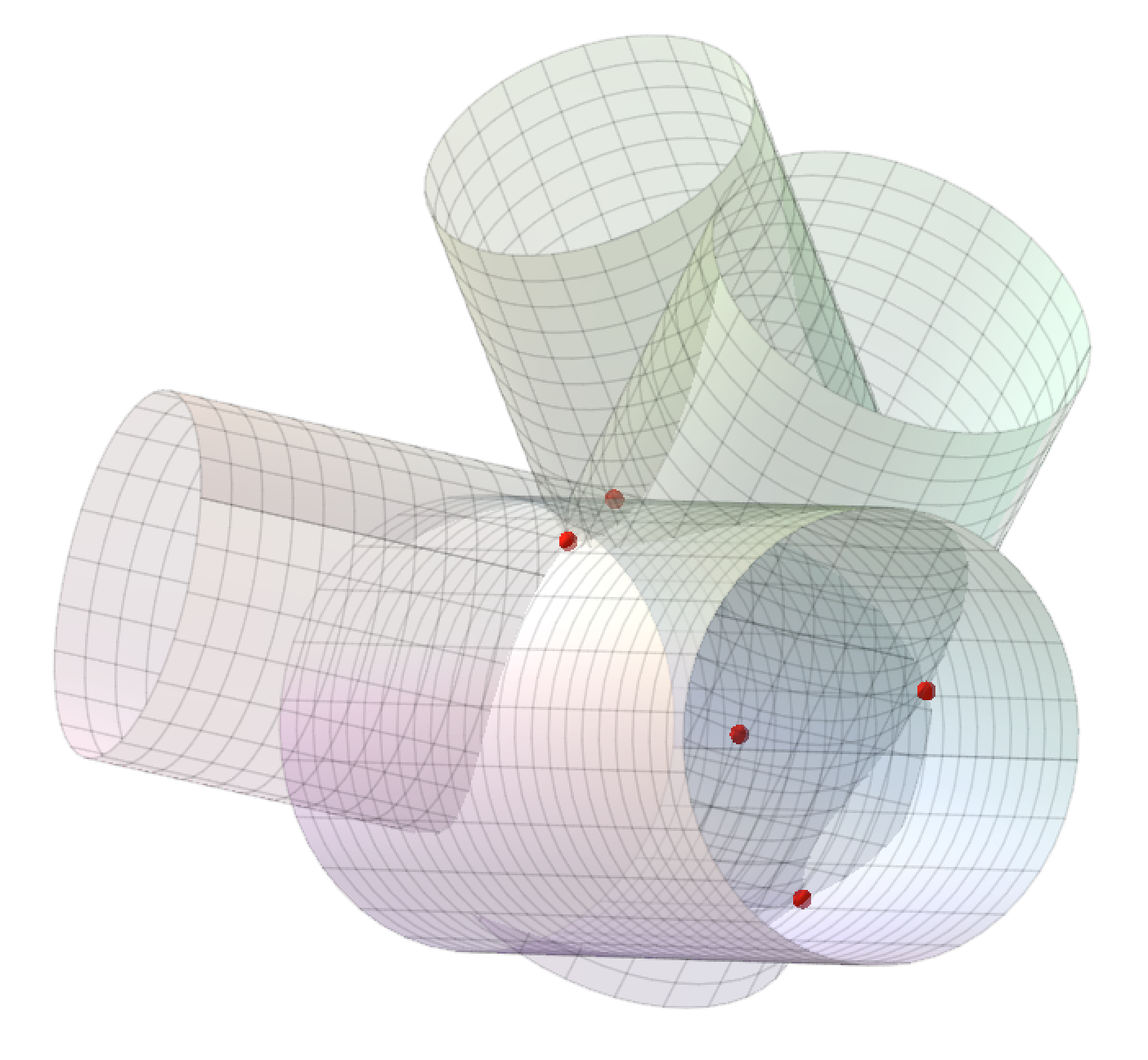}
   \hspace{4em}
   \includegraphics[width=60mm]{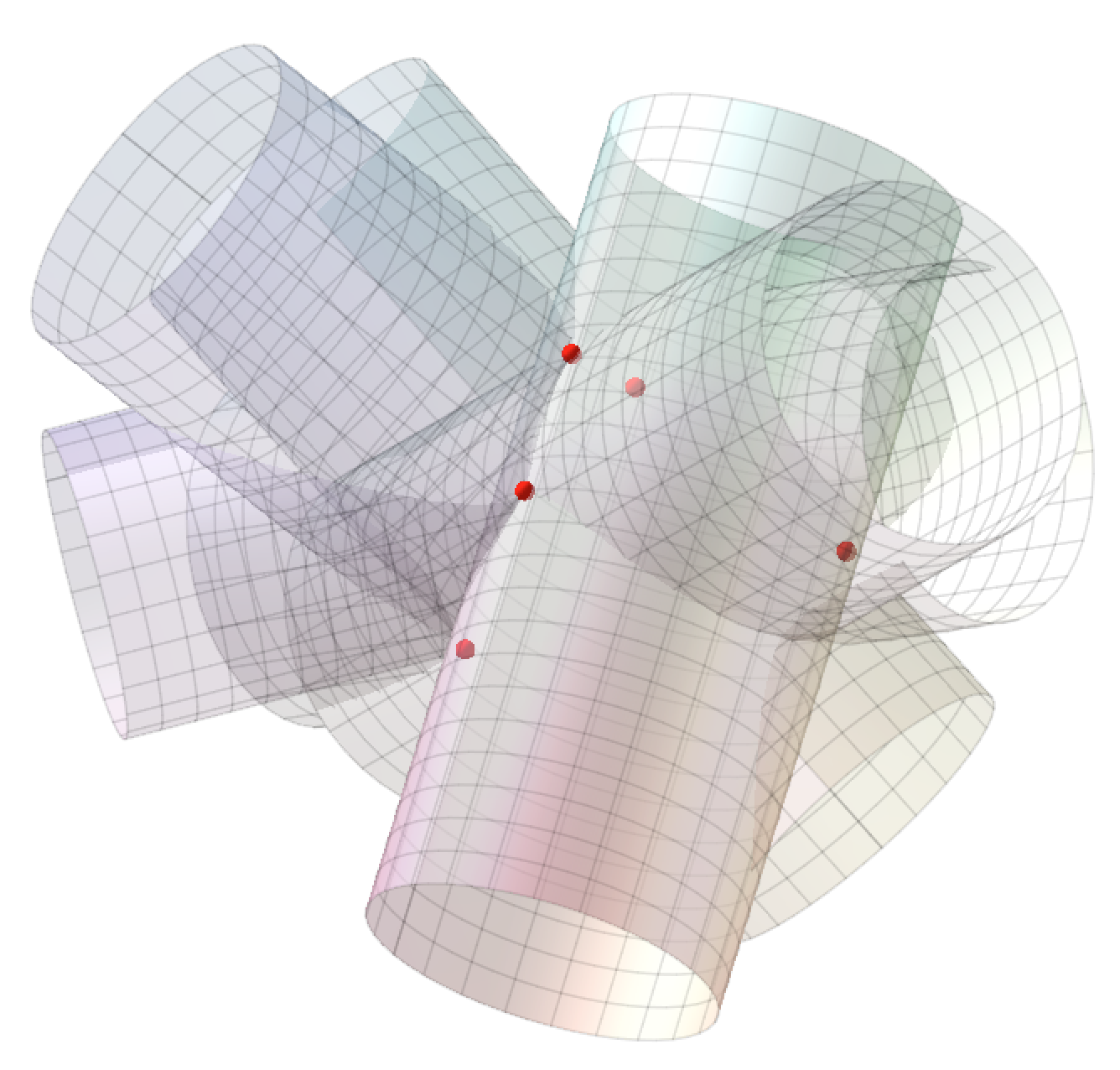}
   \caption{Three configurations of $\Pc$ with 2,4 and 6 cylinders passing through $\Pc$.}\label{fig:cyl5P}
\end{figure}

\begin{table}[ht!]
\begin{center}
\begin{tabular}{c|c|c|c|c|}
 Number of cylinders & 0 & 2 & 4 & 6 \\ 
	 \hline
Proportion (\%)	 &  22.7 & 53.9 & 21.4 & 2 \\
\hline	
\end{tabular}
\caption{Proportion of the number of cylinders found through a thousand random point sets.}\label{tab:cylinders5P}
\end{center}
\end{table}

\section{Interpolation of cones}	

A cone (more precisely a right circular cone) is a three-dimensional geometric shape that consists of the locus of all straight lines (the generatrices) joining a fixed point, called the \emph{apex} to the points of a circle, the apex lying on the line, called the \emph{axis}, passing through the center of the circle and normal to the plane containing this circle. A cone is determined by six parameters : four parameters for the axis, an additional parameter for the apex, which is a point on the axis and a last parameter for the angle of the cone, that is to say the angle made at the apex between the axis and any of the straight lines generating the cone.  

\medskip

\com{This approach} is classical to determine a cone from three oriented points. Indeed, each normal plane to each oriented point must contains the apex, so the apex can be determined as the intersection of three such planes. Nevertheless, this approach is overdetermined in the sense that it will not yield a cone in general \com{because} three general oriented points give nine conditions whereas a cone has only six parameters. Hereafter, we provide methods for interpolating a cone through a mixed point set that define exactly six conditions. There are three possibilities \com{for such} mixed point sets : either two oriented points, or one oriented point and three other distinct points, or six distinct points.

\subsection{Cones through two oriented points} 

We seek for cones through a set $\Pc$ of two oriented points that we will denote by $p_1,n_1$ and $p_2,n_2$. By a linear change of coordinates, one can assume without loss of generality that $p_1:=(0,0,0)$ and  $n_1:=(0,0,1)$. For the remaining data, we set $p_2=(x_2,y_2,z_2)$ and $n_2=(a_2,b_2,c_2)$.

\medskip

Assuming that $n_2$ is not proportional to $n_1$, i.e.~$a_2\neq 0$ or $b_2\neq 0$, the normal plane to $n_1$ through $p_1$ and the normal plane \com{$\Pi_2$} to $n_2$ through $p_2$ intersect along a line $\Lc$. A key remark is that the apex of a cone through $\Pc$ must be on $\Lc$. This line is defined by the equations 
$$\Lc : \left\{\begin{array}{l}
z=0\\
a_2x+b_2y+c_2z-p_2.n_2=0.	
\end{array}\right.
$$
If $n_2$ is proportional to $n_1$, then we get a \com{degenerate} situation. Indeed, either $z_2\neq 0$ and there is no cone through $\Pc$, \com{or} $z_2=0$ and there are infinitely many cones through $\Pc$ (the line through $p_1$ and $p_2$ define a generatrix of the cone and then the apex can be freely chosen on this line, as well as the angle of the cone). \com{So from} now on we  assume that $n_2$ is not proportional to $n_1$, i.e.~$a_2^2+b_2^2>0$.

\medskip 

The axis of a cone through $\Pc$ must intersect the line $\Lc_1$ through $p_1$ with direction $n_1$ and also intersect the line $\Lc_2$ through $p_2$ with direction $n_2$. So to characterize this axis, we pick a point on each of these lines, say $q_1=\lambda_1 n_1 \in \Lc_1$ and $q_2=p_2+\lambda_2n_2 \in \Lc_2$ where $\lambda_1, \lambda_2 \in \RR$, and we consider the line $\Ac$ through $q_1$ and $q_2$. The only special case when $\Ac$ is not well defined is when 
$\Lc_1$ and $\Lc_2$ intersect and both $q_1$ and $q_2$ are located at this intersection point. For the sake of clarity, we will treat this case separately and we assume for the moment that $\Lc_1$ and $\Lc_2$ do not intersect. 

The line $\Ac$ will be the axis of a cone through $\Pc$ providing the two following conditions hold: 
\begin{itemize}
	\item $\Ac$ \com{intersects} $\Lc$, in which case the intersection point is the apex of the cone,
	\item the angle between $n_1$ and $\Ac$ must be equal to the angle between $n_2$ and $\Ac$.
\end{itemize}
$\Ac$ can be parameterized as the set of points $(1-\lambda)q_1+\lambda q_2$ for $\lambda \in \RR$. It is necessary to have $z_2+\lambda_2c_2-\lambda_1\neq 0$ for otherwise $\Ac$ is parallel to the plane $z=0$ and hence it does not intersect $\Lc$, unless $\lambda_1=z_2+\lambda_2c_2=0$ but in this case the apex would be $p_1$ which is impossible (normal vector is not well defined at the apex). So, assuming $z_2+\lambda_2c_2-\lambda_1\neq 0$, the intersection point between $\Ac$ and the plane $z=0$ is the \com{point with} coordinates
\begin{equation}\label{eq:linesInter}
	\left(-\lambda_1\frac{x_2+\lambda_2a_2}{z_2+\lambda_2c_2-\lambda_1} , -\lambda_1 \frac{y_2+\lambda_2b_2}{z_2+\lambda_2c_2-\lambda_1}, 0\right).
\end{equation}
\com{Substituting this point into $\Lc$, we} get the following first constraint on $\lambda_1$ and $\lambda_2$ : 
\begin{equation}\label{eq:cone2P-E1}
	\lambda_1\lambda_2(a_2^2+b_2^2)-\lambda_1z_2c_2 + (p_2.n_2)(z_2+\lambda_2c_2)=0.
\end{equation}
\com{Now denote} by $q$ the vector $q_2-q_1$\com{;} the second condition is obtained by imposing 
$q.n_1=\pm q.n_2/\|n_2\|,$ the sign change \com{arises because} the normals are supposed to be oriented with respect to the interpolating cone. More explicitly, we obtain the equation
\begin{equation}\label{eq:cone2P-E2}
	\rho(z_2+\lambda_2c_2-\lambda_1) = \varepsilon (-c_2\lambda_1+\rho^2\lambda_2+p_2.n_2),
\end{equation}
where $\rho=\|n_2\|=\sqrt{a_2^2+b_2^2+c_2^2}$ and $\varepsilon=\pm 1$. 
This latter condition being linear in $\lambda_1$ and $\lambda_2$, one can solve the system of equations $\eqref{eq:cone2P-E1}$ and $\eqref{eq:cone2P-E2}$ in closed form. More precisely, once $\varepsilon$ is fixed, we get the \com{following two} solutions~:  
  \begin{equation*}
    \left\{
    \begin{aligned}
    &\lambda_1' = \frac{a_2^2z_2-a_2c_2x_2+b_2^2z_2-b_2c_2y_2}{a_2^2+b_2^2}\\
    &\lambda_2' = -\frac{a_2x_2+b_2y_2}{a_2^2+b_2^2}\\
    \end{aligned}
    \right., \ \ 
    \left\{
    \begin{aligned}
    &\lambda_1 = \frac{p_2.n_2}{c_2 + \varepsilon \rho}\\
    &\lambda_2 = - \frac{z_2}{c_2 +\varepsilon \rho }\\
    \end{aligned}
    \right..
  \end{equation*}
Observe that $|c_2|\neq \rho$ since it is assumed that $n_2$ is not parallel to $n_1$. Another important observation is that the solution ($\lambda_1',\lambda_2'$) is independent of $\varepsilon$ and  moreover satisfies $z_2+\lambda_2'c_2-\lambda_1'=0$. Therefore, \com{($\lambda_1',\lambda_2'$)} is not a valid solution. \com{So we} are left with two solutions to our geometric interpolation problem, namely 
 \begin{equation}\label{eq:conemixedsol}
   \left\{
   \begin{aligned}
   &\lambda_1 = \frac{p_2.n_2}{c_2 - \rho}\\
   &\lambda_2 = - \frac{z_2}{c_2 - \rho }\\
   \end{aligned}
   \right., \ \
   \left\{
   \begin{aligned}
   &\lambda_1 = \frac{p_2.n_2}{c_2 +  \rho}\\
   &\lambda_2 = - \frac{z_2}{c_2 + \rho }\\
   \end{aligned}
   \right..
 \end{equation}
Once such a solution is chosen, then one can determine a unique cone through $\Pc$ : its apex is given by \eqref{eq:linesInter}, its direction is given by $q=q_2-q_1$ and its angle is given by the angle between the vectors $p_1-apex$ and $q$.

\medskip

It remains to treat the case where the lines $\Lc_1$ and $\Lc_2$ intersect. Let us denote by $\Pi$ the plane that contains these two lines, by $\Nc_1$ the normal line in $\Pi$ to $\Lc_1$ through $p_1$, by $\Nc_2$ the normal line in $\Pi$ to $\Lc_2$ through $p_2$ and by $\omega$ the intersection point between $\Nc_1$ and $\Nc_2$. We notice that $\omega$ is nothing but the intersection point between $\Pi$ and $\Lc$; see Figure \ref{fig:schema} for a geometric illustration.

First, we observe that there are always two cones throu\-gh $\Pc$. These two cones have the same apex $\omega$ and their \com{axes} are the perpendicular bisectors of the lines $\Nc_1$ and $\Nc_2$ through $\omega$. Therefore, these two cones are symmetric with respect to the plane $\Pi$ and their intersection is composed of the two lines $\Nc_1\cup \Nc_2$. \com{There are} no other cones through $\Pc$ \com{whose} axis is contained in $\Pi$ (equivalently whose apex is $\omega$).

\lb{Next we look} for cones through $\Pc$ \lb{whose} apex $\alpha$ is such that $\alpha \in \Lc \setminus \{\omega\}$.  Since the axis of such a cone must intersect $\Lc_1$ and $\Lc_2$ and go through $\alpha$ which is not in $\Pi$, then \com{the axis} must be the line through $\alpha$ and $q:=\Lc_1\cap \Lc_2 \in \Pi$ (recall that it is assumed that $n_2$ and $n_1$ are not proportional). 
\com{The two triangles $\alpha q p_1$ and $\alpha q p_2$ are right angled at $p_1$, respectively at $p_2$, since $\alpha$ belongs to $\Lc=\Pi_2\cap\{z=0\}$. Therefore, a necessary and sufficient condition for the existence of a cone through $\Pc$, and whose axis is the line through $\alpha$ and $q$, is that the angles $(\alpha q, \alpha p_1)$ and $(\alpha q, \alpha p_2)$ are the same. This condition is equivalent to the equality $\|qp_1\|=\|qp_2\|$ since the two right triangles $\alpha q p_1$ and $\alpha q p_2$ share the edge $\alpha q$.
}


 \begin{figure}
	\hspace{-1.5em}
	\includegraphics[width=100mm]{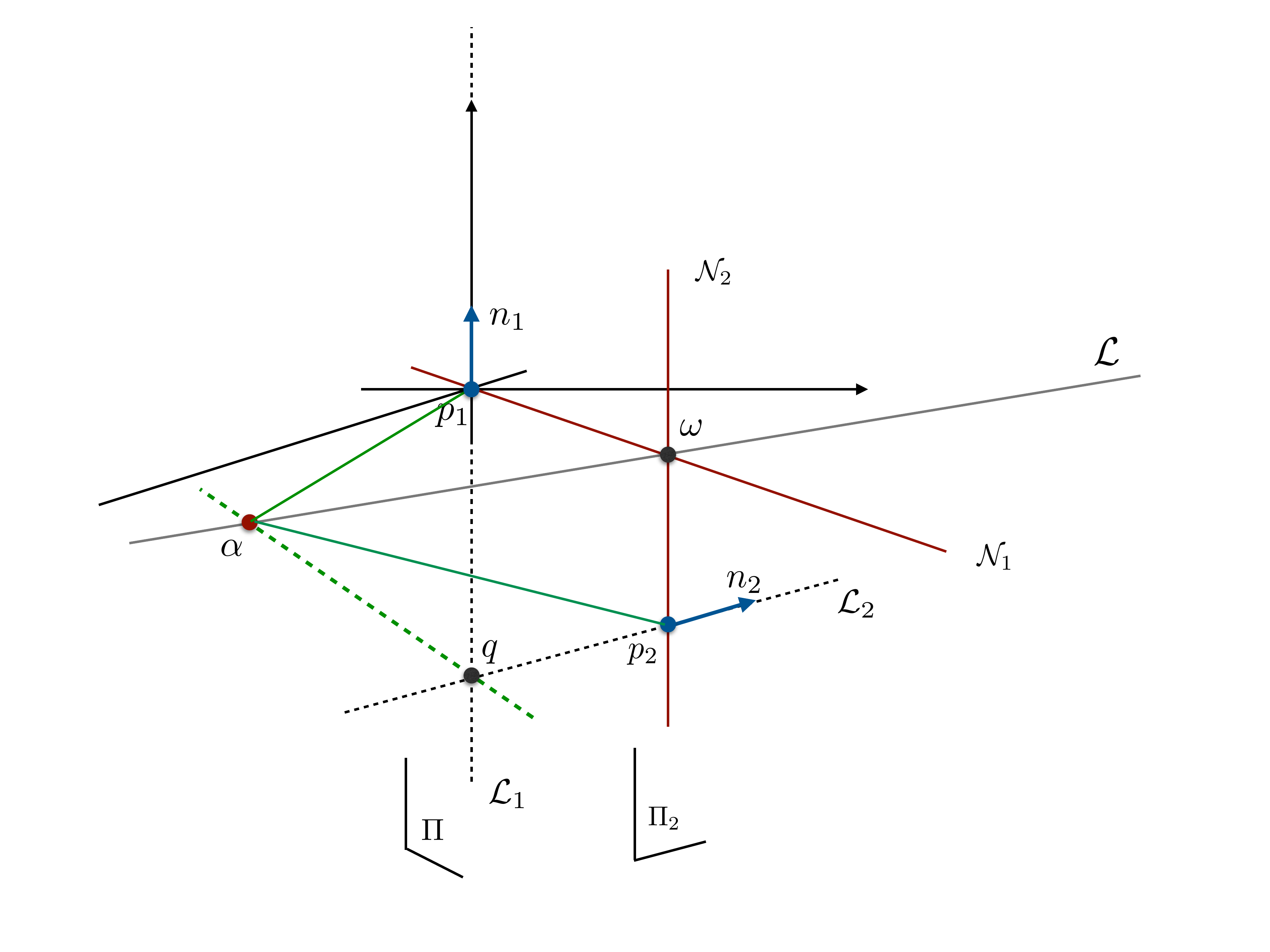}
    \caption{Geometric construction in the case where the lines $\Lc_1$ and $\Lc_2$ intersect.}\label{fig:schema}
 \end{figure}

In summary, when the lines $\Lc_1$ and $\Lc_2$ intersect there always exists two cones through $\Pc$, whose axis are contained in $\Pi$. If $\|qp_1\|\neq \|qp_2\|$ then there are no more cones through $\Pc$, otherwise there are infinitely many cones through $\Pc$, more precisely a 1-dimensional family of cones which is parameterized by $\Lc$. A last observation is that if $\|qp_1\|=\|qp_2\|$ then \com{one of the two} perpendicular bisectors of $\Nc_1$ and $\Nc_2$ through $\omega$ goes through $q$, so that there is a continuity in this family of cones through $\Pc$ when $\alpha$ is moving on the line $\Lc$.

\begin{thm} \lb{Given} an oriented set of two distinct points \lb{$\Pc$, denote} by $\Lc$ the line through $p_1$ with direction $n_1$ and by $\Lc_2$ the line through $p_2$ with direction $n_2$.
	If the \lb{following two} conditions \lb{hold~:}
	\begin{itemize}
		\item[i)] $n_1$ and $n_2$ are not proportional, 
		\item[ii)] if the lines $\Lc_1$ and $\Lc_2$ intersect at a point $q$ then $\|qp_1\|\neq \|qp_2\|$,
	\end{itemize}
then there are exactly two real cones through $\Pc$. Moreover, the intersection of these two cones consists only of the points $p_1$ and $p_2$, except if they share the same apex in which case their intersection consists of the lines $\Nc_1$ and $\Nc_2$.  
	
	If $n_1$ is parallel to $n_2$ then there is no cone through $\Pc$ unless $p_2$ belongs to the \lb{line} $\Nc_1$, in which case there are infinitely many cones through $\Pc$.
	
	If the lines $\Lc_1$ and $\Lc_2$ intersect at a point $q$ such that $\|qp_1\|=\|qp_2\|$, then there are infinitely many cones through $\Pc$.
\end{thm}

\paragraph{Extraction algorithm}

From the above analysis, we deduce the following extraction algorithm \lb{for cones} through a set of two oriented points $\Pc$.
\begin{enumerate}
	\item Perform a linear change of coordinates \lb{so that} $p_1=(0,0,0)$ and $n_1=(0,0,1)$. 
	\item If $n_1$ and $n_2$ are proportional then check \lb{whether} $p_2$ belongs to $\Nc_1$ and return that either there is no cone through $\Pc$, or there are infinitely many cones through $\Pc$.
	\item If $\Lc_1$ \lb{intersects} $\Lc_2$, then if $\|qp_1\|\neq \|qp_2\|$ return the two cones through $\Pc$, otherwise if $\|qp_1\|= \|qp_2\|$ return that there are infinitely many cones through $\Pc$. 
	\item Now, since $n_1$ is not proportional to $n_2$ and $\Lc_1\cup\Lc_2=\emptyset$, there are exactly two cones through $\Pc$ that are computed by means of the closed formulas \eqref{eq:conemixedsol}. 
\end{enumerate} 

This algorithm has been implemented with the {\tt Maple} software. We observed that \com{computing} the \com{cones} through a random set of points takes in average 3.8ms \com{(including all the steps in the above algorithm)} and is almost constant, \com{i.e.~independent} of the point set. Some illustrative configurations are shown in Figure \ref{fig:cone2N-1} and Figure \ref{fig:cone2N-2}.

 \begin{figure}
 	\centering
	\includegraphics[width=50mm]{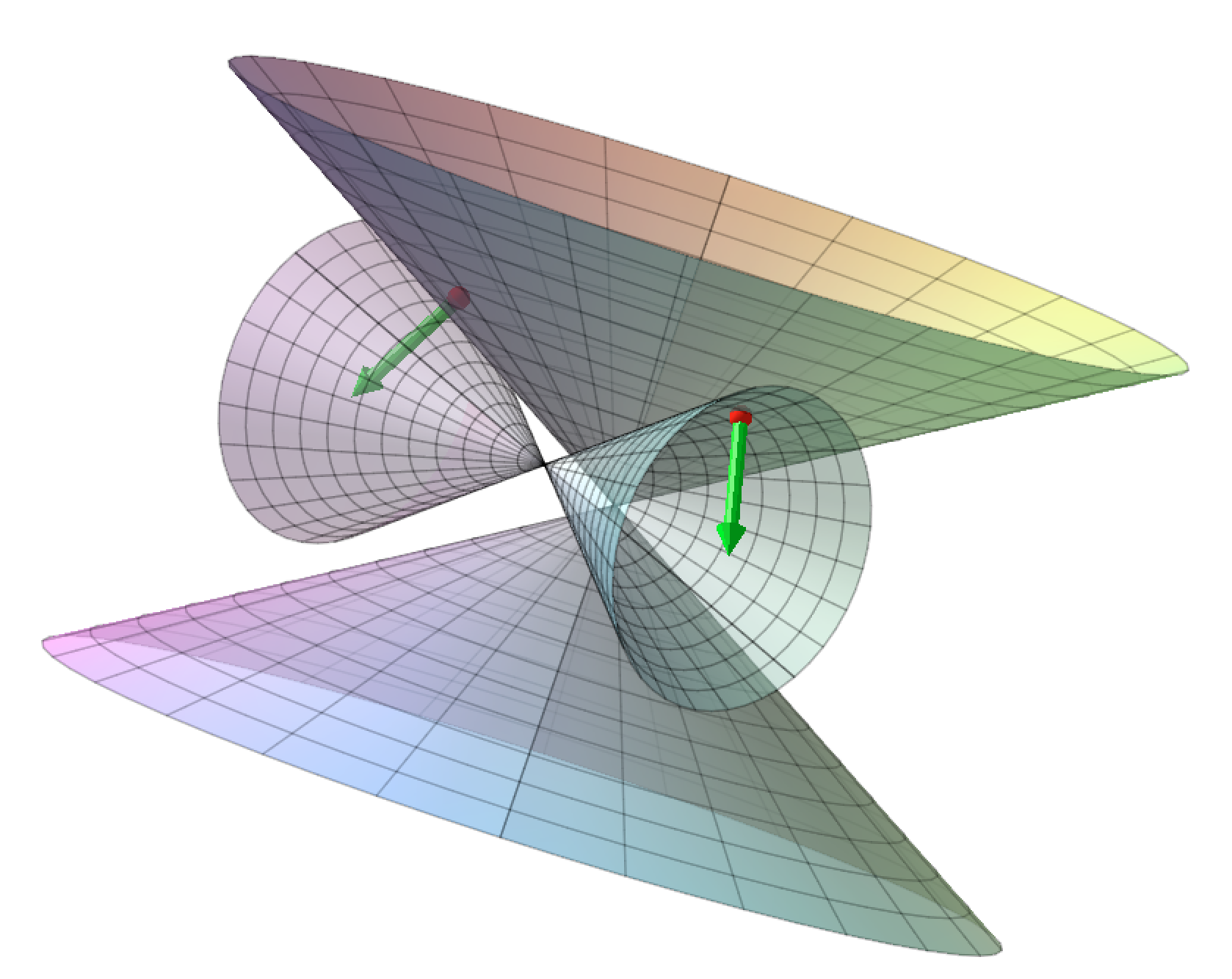}
    \caption{Two cones through a general set of two oriented points.}\label{fig:cone2N-1}
 \end{figure}

  \begin{figure}
 	 \centering
	 	 	     \includegraphics[width=60mm]{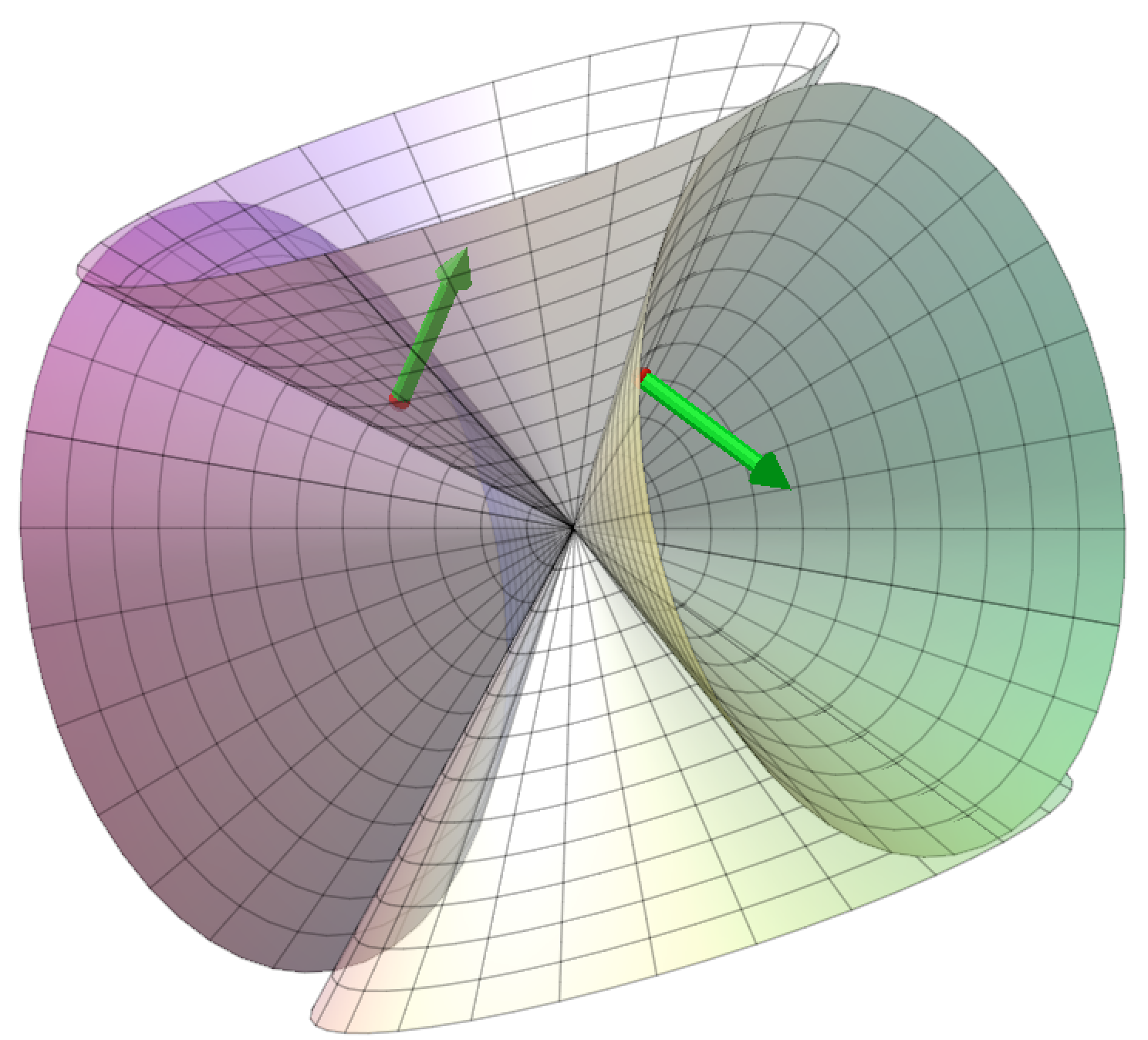}
    \caption{Two cones intersecting along two lines and sharing the same apex.}\label{fig:cone2N-2}
 \end{figure}

\subsection{Cones through a mixed minimal point set}
 \com{We seek} cones through a set $\Pc$ \com{that} is composed of an oriented point $(p_1,n_1)$ and three other simple points $p_2,p_3,p_4$. By a linear change of coordinates, one can assume without loss of generality that $p_1:=(0,0,0)$ and  $n_1:=(0,0,1)$. For the remaining data, we set $p_2=(x_2,y_2,z_2)$ , $p_3=(x_3,y_3,z_3)$ and $p_4=(x_4,y_4,z_4)$.
 
\medskip

The axis of a cone through $\Pc$ must intersect the line generated by $n_1$ through $p_1$ in a point $q=(0,0,r)$. Observe that $r$ must be nonzero otherwise $p_1$ would be the apex of such a cone, which is impossible since it is a nonsingular point. In addition, the apex of a cone through $\Pc$ must belong to the normal plane to $n_1$ through the point $p_1$ and we choose to represent the coordinates of this point by $p=(ar,br,0)$. Again, observe for the same reason that $a$ and $b$ must be nonzero, i.e.~$a^2+b^2\neq 0$.

\com{Now} a point $m=(x,y,z)$ belonging to a cone through $\Pc$ must satisfy
$$ 
(\vec{pm}\cdot\vec{pq})^2=\| pm\|^2 \| pq\|^2  \kappa,
$$
that is to say
\begin{multline*}
r^2\left(-a(x-ar) -b(y-br) -z\right)^2 = \\ r^2\left( (x-ar)^2+ (y-br)^2 +z^2 \right) 
\left( a^2 + b^2 + 1 \right)\kappa
\end{multline*}
where $\kappa>0$ is a constant that corresponds to the angle of the cone.  The factor $r^2$ can be cancelled out. Moreover, the cone goes through $p_1$ so we deduce that
$$ r^2\left(a^2+b^2 \right)^2 = r^2\left( a^2+ b^2  \right)
 \left( a^2 + b^2 + 1 \right)\kappa.
$$
\com{Substituting into} the previous equation, we finally deduce that the point $m=(x,y,z)$ belonging to a cone through $\Pc$ must satisfy $F(x,y,z)=0$ where
\begin{multline*}
	F(x,y,z)  =2\left( {a}^{2}+{b}^{2} \right) zr+ \left( {a}^{2}+{b}^{2}-1
 \right) {z}^{2} \\ -2\left( ax+by \right) z+ \left( bx-ay \right) ^{2
}.
\end{multline*}
The three parameters $a,b$ and $r$ are to be determined and we will use the three points $p_2,p_3$ and $p_4$ for that purpose. 

Assume that \com{one of the points} $p_2,p_3,p_4$, say $p_2$, is not in \com{the plane $z=0$} (the case where all the points are in this plane \com{is a special case of the situation} discussed \com{at the end of Section \ref{sec:C6P}}). Hence $z_2\neq 0$. \com{Then} there exists a cone through $\Pc$ if and only if 
$$ F(x_2,y_2,z_2) = 0, \ \ z_3F(x_2,y_2,z_2)-z_2F(x_3,y_3,z_3)=0,$$ $$z_4F(x_2,y_2,z_2)-z_2F(x_4,y_4,z_4)=0.$$  
It turns out that the two last equations are independent of $r$; for $i=3,4$ we have 
\begin{multline*}
z_iF(x_2,y_2,z_2)-z_2F(x_i,y_i,z_i) = \\
\left( {{ y_2}}^{2}{ z_i}-{{ y_i}}^{2}{ z_2}+{{ z_2}}^{2}{
 z_i}-{ z_2}{{ z_i}}^{2} \right) {a}^{2} 
 - \left( 2{ x_2}
{ y_2}{ z_i}-2{ x_i}{ y_i}{ z_2} \right) ab
\\ - \left( 2{ x_2}{ z_2}{ z_i}-2{ x_i}{ z_2}{ 
z_i} \right) a +
\left( {{ x_2}}^{2}{ z_i}-{{ x_i}}^{2}{ z_2}+{{
 z_2}}^{2}{ z_i}-{ z_2}{{ z_i}}^{2} \right) {b}^{2} \\
 - \left( 
2{ y_2}{ z_2}{ z_i}-2{ y_i}{ z_2}{ z_i}
 \right) b-{{ z_2}}^{2}{ z_i}+{ z_2}{{ z_i}}^{2}.	
\end{multline*}
Since these two equations are quadratic in $a$ and $b$, they have four common roots. Once $a$ and $b$ are determined then $r$ is uniquely defined by the equation  $F(x_2,y_2,z_2) = 0$ which is linear in $r$. So, we have proved the following result. 

\begin{thm} \com{Given} a general mixed point set $\Pc$, \com{there are} 0, 2 or 4 cones through $\Pc$. 
\end{thm}

\paragraph{Extraction algorithm}

From the above analysis, the extraction algorithm \com{for cones} through a mixed point set $\Pc$ composed of an oriented point and three other distinct points relies \com{on solving} two algebraic equations in two variables. To solve such systems, we proceed as in Section \ref{subsec:cyl5P} by means of eigen-computations from a pencil of matrices which is in closed form in terms of the input data. This pencil $A,B$ of $8\times 8$ matrices is obtained as the companion matrices of both equations $z_iF(x_2,y_2,z_2)-z_2F(x_i,y_i,z_i)=0$, $i=3,4$.

\begin{enumerate}
	\item Perform a linear change of coordinates \com{so that} $p_1=(0,0,0)$ and $n_1=(0,0,1)$. 
	\item Instantiate the pencil of matrices $A,B$ with the coordinates of the input points.
	\item Compute the eigenvalues and eigenvectors and deduce the real cones through $\Pc$.
\end{enumerate} 

This algorithm has been implemented with the {\tt Maple} software. We observed that \com{computing} the \com{cones} through a random set of points takes in average 7.5ms \com{(including all the steps in the above algorithm)} and is almost constant,  \com{i.e.~independent} of the point set. The proportion of the number of \com{cones} found through random point sets is given in Table \ref{tab:cone1N3P} and some illustrative configurations are shown in Figure \ref{fig:cone1N3P}.

\begin{figure}[ht!]
\centering	
   \includegraphics[width=65mm]{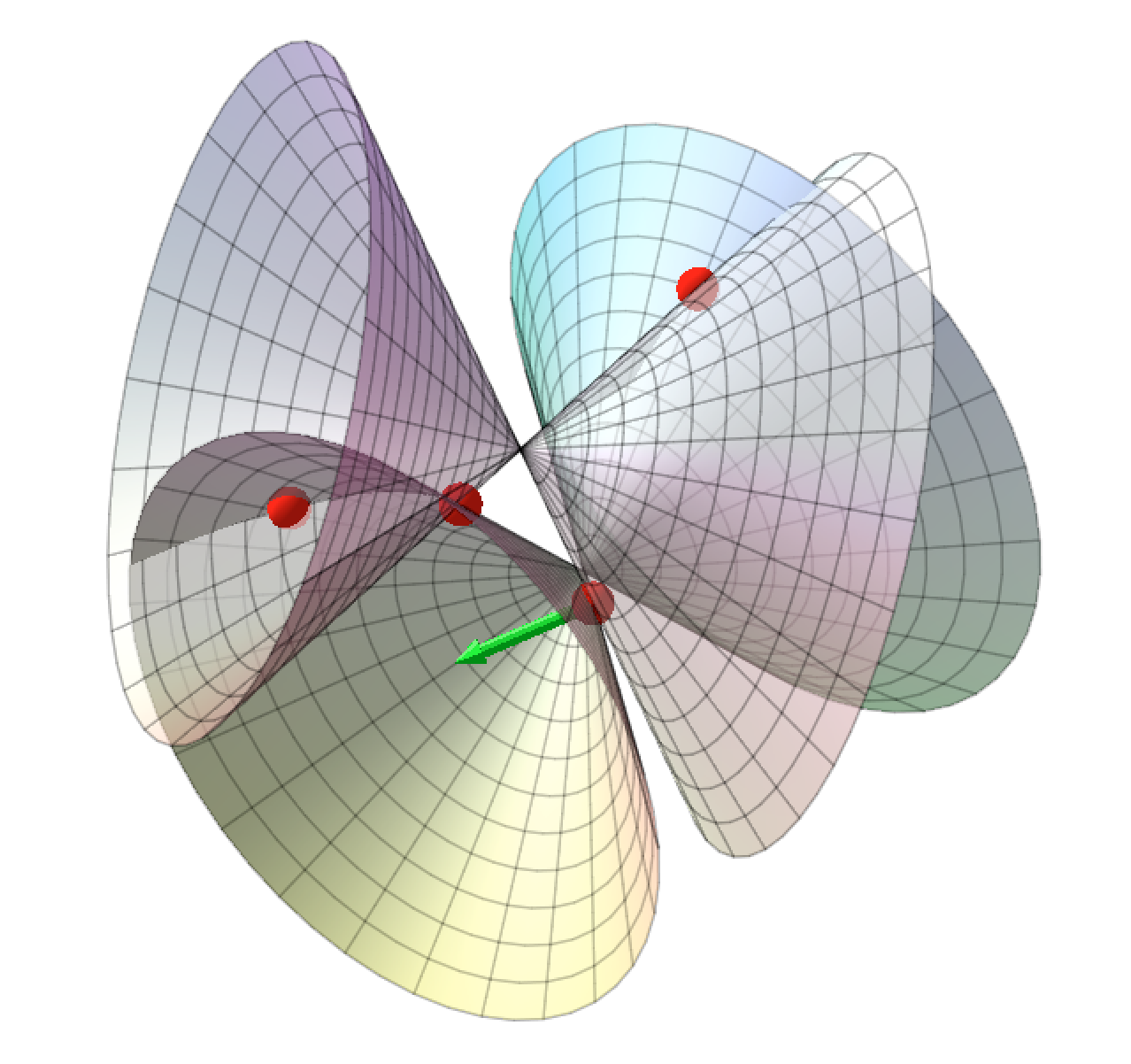}
   \hspace{4em}
   \includegraphics[width=65mm]{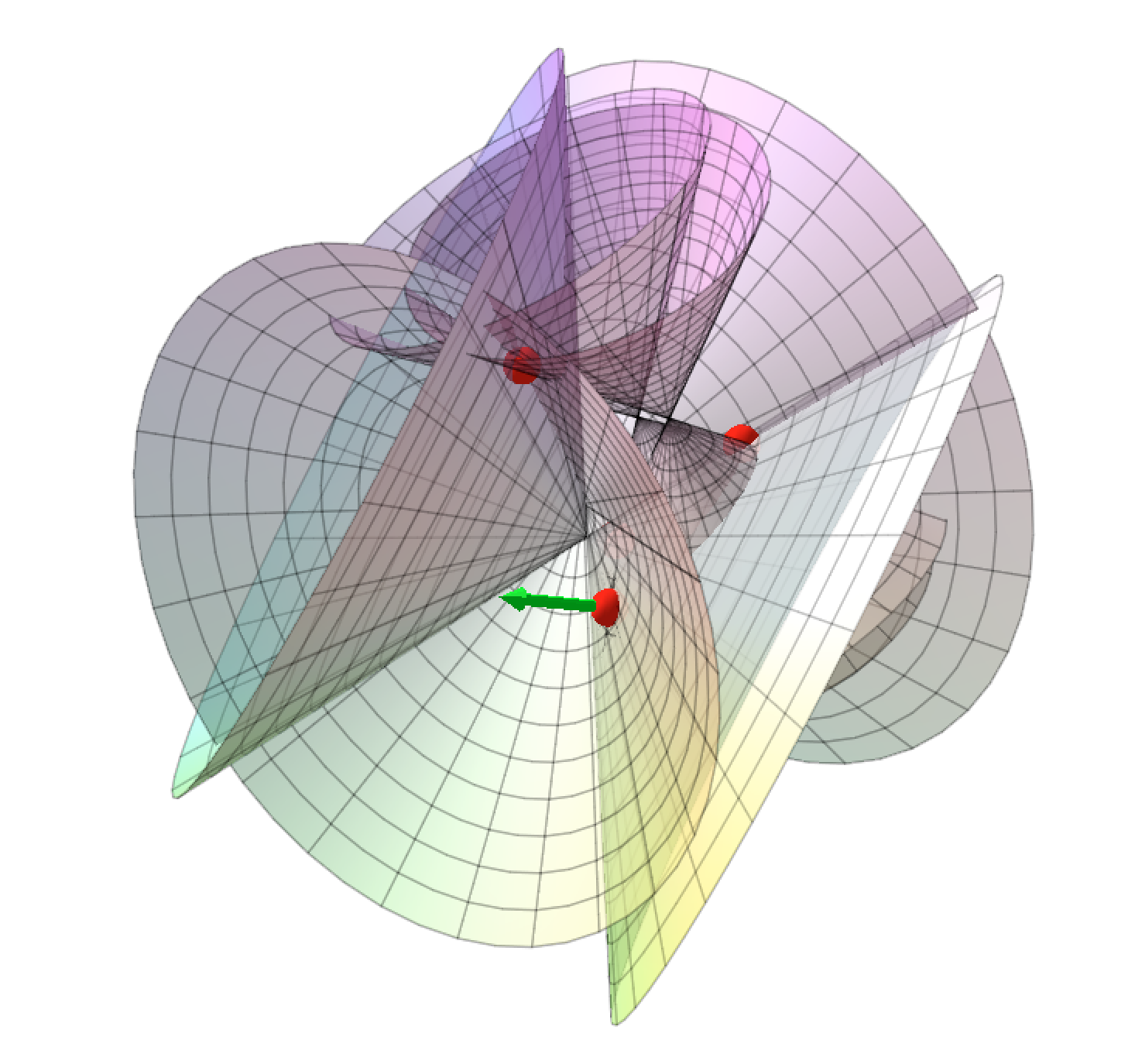}
   \caption{Two configurations of $\Pc$ with 2 and 4 cones passing through $\Pc$.}\label{fig:cone1N3P}
\end{figure}  

\begin{table}[ht!]
\begin{center}
\begin{tabular}{c|c|c|c|}
 Number of cones & 0 & 2 & 4 \\ 
	 \hline
Proportion (\%)	 &  6.9 & 85.9 & 7.2 \\
\hline	
\end{tabular}
\caption{Proportion of the number of \com{cones} found through a thousand random point sets.}\label{tab:cone1N3P}
\end{center}
\end{table}

\subsection{Cones through six points}\label{sec:C6P}
\com{We seek} cones through a set $\Pc$ composed of six simple points $p_1,p_2,p_3,p_4,p_5,p_6$. By a linear change of coordinates, we can assume without loss of generality that
\begin{align}\label{eq:hyp-p1p2p3p4p5p6}
& p_1=(0,0,0), \ p_2=(x_2,0,0), \ p_3=(x_3,y_3,0), \\ \nonumber
& p_4=(x_4,y_4,z_4), \ p_5=(x_5,y_5,z_5), \ p_6=(x_6,y_6,z_6).
\end{align} 

\paragraph{Geometric analysis}
The general equation of a cone depends on 6 parameters:
\begin{multline}\label{eq:eqcone}
F(x,y,z):= \\ (x-a)^2+(y-b)^2+(z-c)^2-(l(x-a)+m(y-b)+n(z-c))^2	
\end{multline}
where $(a,b,c)$ are the coordinates of the apex and the vector $(l,m,n)$ encodes the direction of the axis as well as the angle of the cone ($l^2+m^2+n^2=1/\cos^2(\theta)$, $\theta$ being the angle of the cone).

We introduce the quantity $k:=al+bm+cn$. Since $p_1$ is the origin, we get the equation
\begin{equation}\label{eq:F0}
F(p_1)=a^2+b^2+c^2-k^2=0.	
\end{equation}
Moreover, we observe that the polynomial
\begin{multline}\label{eq:F-F1}
F(x,y,z)-F(p_1)=
 2\left( lx+my+nz \right) k -2(ax+by+cz) \\
 -2lmxy-2lnxz-2mnyz
 -{l}^{2}{x}^{2}-{m}^{2}{y}^{2}-{n}^{2}{z}^{2}+{x}^{2}+{y}^{2}+{z}^{2}	
\end{multline}
is linear in $a,b,c$ and $k$. From the equation
$$F(p_2)-F(p_1)=
2l{ x_2}k-2{ x_2}a-{l}^{2}{{ x_2}}^{2}+{{ x_2}}^{2}=0$$
we get, assuming $x_2\neq 0$, i.e.~$p_2\neq p_1$,
\begin{equation}\label{eq:a}
a=lk+ \frac{x_2}{2}(1-l^2).	
\end{equation}
Now we have that
\begin{multline*}
F(p_3)-F(p_1)= 
2m{y_3}k-2{y_3}b 
+{x_3}{l}^{2}{x_2} \\ -{l}^{2}{{x_3}}^{2}-2lm{x_3}{y_3}
-{m}^{2}{{y_3}}^{2}-{x_3}{x_2}
+{{x_3}}^{2}+{{y_3}}^{2} =0	
\end{multline*}
and we get, assuming $y_3\neq 0$, i.e.~assuming that $p_1,p_2$ and $p_3$ are not aligned,
\begin{multline}\label{eq:b}
b=
mk+ \frac{1}{2y_3}
\left(
{x_3}{l}^{2}{x_2}-{l}^{2}{{x_3}}^{2}-2lm{x_3}{y_3} \right. \\
\left. -{m}^{2}{{y_3}}^{2}-{x_3}{x_2}+{{x_3}}^{2}+
{{y_3}}^{2}
\right).	
\end{multline}
Continuing this way, we get the equation
\begin{multline*}
F(p_4)-F(p_1)= 2 { z_4}nk -2 { z_4} c  \\
+ \frac{1}{y_3}(
-{l}^{2}{ x_2} { x_3} { y_4}+{l}^{2}{ x_2} { x_4} { 
y_3}+{l}^{2}{{ x_3}}^{2}{ y_4}-{l}^{2}{{ x_4}}^{2}{ y_3} \\ 
+2 lm{
 x_3} { y_3} { y_4}-2 lm{ x_4} { y_4} { y_3}-2 ln{
 x_4} { z_4} { y_3} +{m}^{2}{{ y_3}}^{2}{ y_4} \\
 -{m}^{2}{{
 y_4}}^{2}{ y_3}-2 mn{ y_4} { z_4} { y_3}-{n}^{2}{{ z_4
}}^{2}{ y_3} 
+{ x_2}
 { x_3} { y_4}-{ x_2} { x_4} { y_3} \\
 -{{ x_3}}^{2}{ 
y_4}+{{ x_4}}^{2}{ y_3}-{{ y_3}}^{2}{ y_4}+{{ y_4}}^{2}{ 
y_3}+{{ z_4}}^{2}{ y_3})
=0	
\end{multline*}
\com{and we get}, assuming that $z_4\neq 0$, i.e.~assuming that $p_1,p_2,p_3$ and $p_4$ are not coplanar,
\begin{multline}\label{eq:c}
c= nk +
\frac{1}{2y_3z_4} ( 
-{l}^{2}{  x_2} {  x_3} {  y_4}+{l}^{2}{  x_2} {  x_4} {  
y_3}+{l}^{2}{{  x_3}}^{2}{  y_4}-{l}^{2}{{  x_4}}^{2}{  y_3} \\
+2 lm{
  x_3} {  y_3} {  y_4}-2 lm{  x_4} {  y_4} {  y_3}-2 ln{
  x_4} {  z_4} {  y_3} 
  +{m}^{2}{{  y_3}}^{2}{  y_4}-{m}^{2}{{
  y_4}}^{2}{  y_3} \\ -2 mn{  y_4} {  z_4} {  y_3}-{n}^{2}{{  z_4
}}^{2}{  y_3}+{  x_2} {  x_3} {  y_4}-{
  x_2} {  x_4} {  y_3}-{{  x_3}}^{2}{  y_4}+{{  x_4}}^{2}{  
y_3} \\ -{{  y_3}}^{2}{  y_4}+{{  y_4}}^{2}{  y_3}+{{  z_4}}^{2}{  
y_3}
).
\end{multline}
From here, we substitute $a,b$ and $c$ by the above quantity that we found in both equations $F(p_i)-F(p_1)=0$ with $i=5,6$. \com{Notice} that these two equations do not depend on $k$ because of the cancellation of the two terms $ 2\left( lx+my+nz \right) k -2(ax+by+cz)$ in \eqref{eq:F-F1} by this substitution. We get two equations $H_i(l,m,n)$, $i=5,6$, that satisfy
\begin{multline*}
y_3z_4H_i(l,m,n)= \\
	 \left( { x_2} { x_3} { y_4} { z_i}-{ x_2} { x_3} {
	 y_i} { z_4}-{ x_2} { x_4} { y_3} { z_i}+{ x_2} {
	 x_i} { y_3} { z_4} \right. \\ \left. -{{ x_3}}^{2}{ y_4} { z_i}+{{ x_3}
	}^{2}{ y_i} { z_4}+{{ x_4}}^{2}{ y_3} { z_i}-{{ x_i}}^{2
	}{ y_3} { z_4} \right) {l}^{2} \\
	+ \left( -2 { x_3} { y_3} {
	 y_4} { z_i}+2 { x_3} { y_3} { y_i} { z_4}+2 { x_4
	} { y_3} { y_4} { z_i}-2 { x_i} { y_3} { y_i} { 
	z_4} \right) lm \\
	+ \left( 2 { x_4} { y_3} { z_4} { z_i}-2 {
	 x_i} { y_3} { z_4} { z_i} \right) ln \\
	 + \left( -{{ y_3}}^{
	2}{ y_4} { z_i}+{{ y_3}}^{2}{ y_i} { z_4}+{ y_3} {{ 
	y_4}}^{2}{ z_i}-{{ y_i}}^{2}{ y_3} { z_4} \right) {m}^{2} \\
	+
	 \left( 2 { y_3} { y_4} { z_4} { z_i}-2 { y_3} { y_i
	} { z_4} { z_i} \right) mn 
	+ \left( { y_3} {{ z_4}}^{2}{ 
	z_i}-{{ z_i}}^{2}{ y_3} { z_4} \right) {n}^{2}\\ 
	-{ x_2} { x_3
	} { y_4} { z_i}+{ x_2} { x_3} { y_i} { z_4}+{ x_2}
	 { x_4} { y_3} { z_i}-{ x_2} { x_i} { y_3} { z_4} \\ +
	{{ x_3}}^{2}{ y_4} { z_i}-{{ x_3}}^{2}{ y_i} { z_4}-{{
	 x_4}}^{2}{ y_3} { z_i}+{{ x_i}}^{2}{ y_3} { z_4}+{{ 
	y_3}}^{2}{ y_4} { z_i}\\
	-{{ y_3}}^{2}{ y_i} { z_4}-{ y_3} 
	{{ y_4}}^{2}{ z_i}+{{ y_i}}^{2}{ y_3} { z_4}-{ y_3} {{
	 z_4}}^{2}{ z_i}+{{ z_i}}^{2}{ y_3} { z_4}=0.
\end{multline*}
These equations $H_5$ and $H_6$ are polynomials of degree 2 in $l,m,n$ and only the constant term and monomials of degree 2 have nonzero coefficients. The equation $H_5$ depends on the coefficients of $p_2,p_3,p_4$ and $p_5$ whereas $H_6$ depends on the coefficients of $p_2,p_3,p_4$ and $p_6$.   \com{Applying} the same substitutions and replacing $k$ by its defining value in \eqref{eq:F0}, we get another equation $H_0(l,m,n)$ which is of degree 6 in $l,m,n$ and such that only monomials of even degree have nonzero coefficients. Observe that $H_0$ \com{depends only on} the coordinates of $p_2,p_3$ and $p_4$. To give an idea, we have
\begin{multline*}
4y_3^2z_4^2(l^2+m^2+n^2-1)H_0= \\
 \left( -{{ x_2}}^{2}{{ x_3}}^{2}{{ y_4}}^{2}-{{ x_2}}^{2}{{
 x_3}}^{2}{{ z_4}}^{2}+2 {{ x_2}}^{2}{ x_3} { x_4} { 
y_3} { y_4}-{{ x_2}}^{2}{{ x_4}}^{2}{{ y_3}}^{2} \right. \\ \left. +2 { x_2} 
{{ x_3}}^{3}{{ y_4}}^{2} +2 { x_2} {{ x_3}}^{3}{{ z_4}}^{2} 
-2 { x_2} {{ x_3}}^{2}{ x_4} { y_3} { y_4} \right. \\
\left. 
-2 { x_2} 
{ x_3} {{ x_4}}^{2}{ y_3} { y_4}+2 { x_2} {{ x_4}}^{3}
{{ y_3}}^{2}-{{ x_3}}^{4}{{ y_4}}^{2}-{{ x_3}}^{4}{{ z_4}}^{
2}  \right. \\ \left. +2 {{ x_3}}^{2}{{ x_4}}^{2}{ y_3} { y_4}-{{ x_4}}^{4}{{
 y_3}}^{2} \right) {l}^{6} +\cdots 
\end{multline*}
We notice that the quantity $l^2+m^2+n^2-1$ can be assumed to be nonzero because cones such that $l^2+m^2+n^2=1$ \com{degenerate} to straight lines.

The common roots of the equations $H_0,H_5$ and $H_6$ \com{yield} exactly all the cones through $\Pc$. If finite, this number is at most  $2\times 2\times 6=24$ because of the B\'ezout theorem. However, each cone through $\Pc$ yields two solutions of this polynomial system $H_0=0,H_5=0,H_6=0$, namely $(l_0,m_0,n_0)$ and $(-l_0,-m_0,-n_0)$. This latter property is a consequence of the fact that $H_0,H_5,H_6$ \com{can be written as} monomials of even total degree in $l,m,n$. Therefore, we deduce that \com{there are} at most 12 cones through $\Pc$. Nevertheless, in order to devise an efficient algorithm for computing those cones through $\Pc$, we apply a transformation (called a \emph{$\sigma$-process} or a \emph{blow-up} in birational geometry) in order to break the above central symmetry and reduce the degree of the equations so that each cone will correspond to a unique root of the new polynomial system. 

Consider the transformation
\begin{eqnarray*}
	\sigma : \RR^3 & \rightarrow & \RR^3 \\
	(l,m,n) & \mapsto & \left(u:=\frac{l}{n},v:=\frac{m}{n},n\right).
\end{eqnarray*}
Since the equations $H_0, H_5$ and $H_6$ \com{contain only} even degree monomials, their transform under $\sigma$ will be a polynomial in $n^2$ (and not only a polynomial in $n$). More precisely, setting $w:=n^2$, we have
\begin{align*}
	H_0(un,vn,n) & =  w^3P(u,v)+w^2Q(u,v)+wR(u,v)+S, \\
	H_5(un,vn,n) & =  wH(u,v)+C,\\ 
	H_6(un,vn,n) & =  wK(u,v)+D,
\end{align*}
where $H,K,P,Q,R$ are polynomials in $u,v$ of degree 2, 2, 6, 4, 2 respectively, and where $C,D,S$ are \com{constants} that \com{depend only} on the $x_i$'s, $y_i$'s and $z_i$'s. \com{Since} the \com{last two} equations are linear in $w$ we can simplify this polynomial system \com{in $w$}, assuming $w\neq 0$, to get 
\begin{align}\label{eq:cone6Psys}
	E_0(u,v) & =  C^3P(u,v)-C^2Q(u,v)H(u,v)\\ \nonumber
	 & \hspace{5em}+CR(u,v)H(u,v)^2-SH(u,v)^3, \\ \nonumber
	E_1(u,v) & =  DH(u,v)-CK(u,v),\\ \nonumber
	E_2(u,v,w) & =  wK(u,v)+D.
\end{align}
The two equations $E_0$ and $E_1$ are bivariate polynomials in $u,v$ of degree $6$ and $2$ respectively. Therefore, they define at most $12$ solutions. In addition, for each solution the value of $w$ can be computed from $E_2$. As a consequence, we have just proved the following result.

\begin{thm}\label{thm:12cones}
	 \com{Given} a point set $\Pc$ composed of six distinct points \com{then}, if finite, there is an even number of cones through $\Pc$, which is possibly 0 and at most 12.
\end{thm}
\begin{proof} \com{Solving} the polynomial system \com{given by the equations $E_0(u,v)=0$ and $E_1(u,v)=0$} leads to an even number of real solutions with a maximum of $12=6\times 2$ solutions. So the only thing to show is that there is a bijection between these solutions and the cones through $\Pc$. For that purpose, we first observe that the value of $w$ is uniquely determined from each of those roots by equation $E_2$. And then the solutions \com{$(u,v,w)$ uniquely \lb{pull back}} to a solution $(l,m,n)$ under $\sigma$. 
	
	\com{So} it remains to show that once a solution $(l,m,n)$ is computed, then the parameters $a,b$ and $c$ are uniquely determined. This is indeed the case because $a,b$ and $c$ can \com{actually be given} in closed form in terms of $l,m,n$ and the coordinates of the input \com{points because} the three equations \eqref{eq:a}, \eqref{eq:b} and \eqref{eq:c} \com{yield} a linear system in $a,b,c$ after replacing $k$ by its defining value $al+bm+cn$.
\end{proof}

In Figure \ref{fig:cones} it is shown that all the \com{possible solutions} are reached in practice. Finally, observe that the particular case $w=n^2=0$ can be treated independently in the same \com{vein} as above. Indeed, the remaining variables $l,m$ can be computed from the equations $H_5(l,m,0)=0$ and $H_6(l,m,0)=0$. Moreover, to solve efficiently this system one can perform a transformation similar to $\sigma$ by considering the equations $H_5(um,m,0)=0$ and $H_6(um,m,0)=0$ that are linear equations in $m^2$. At the end, the computed solutions \com{must satisfy} the equation $H_0$ in order to validate that they correspond to cones through the given set of points $\Pc$.   

\begin{figure}[ht!]
\centering	
   \includegraphics[width=40mm]{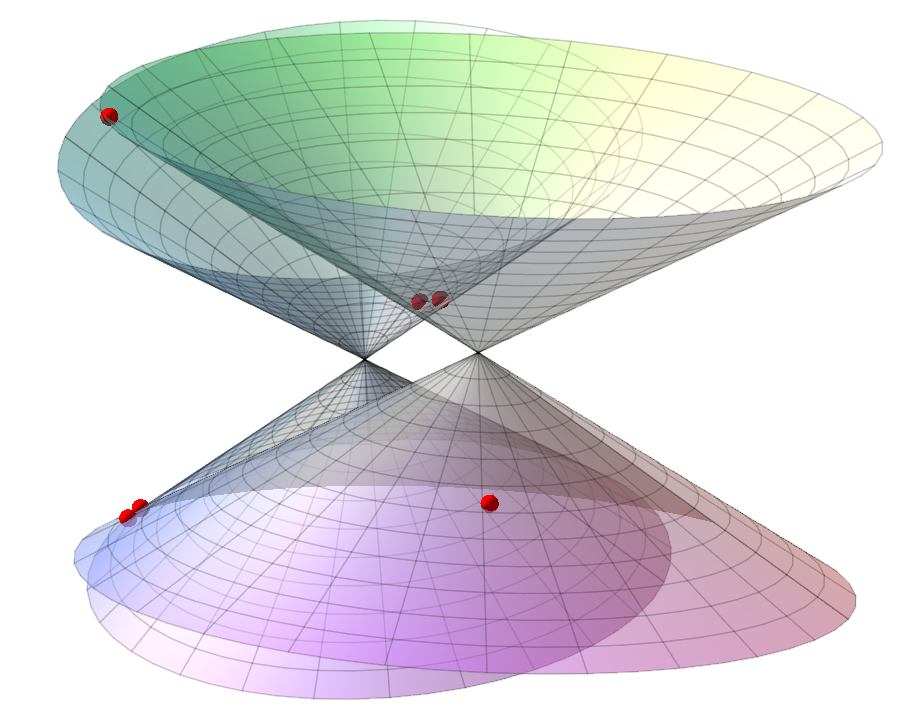}
   \includegraphics[width=40mm]{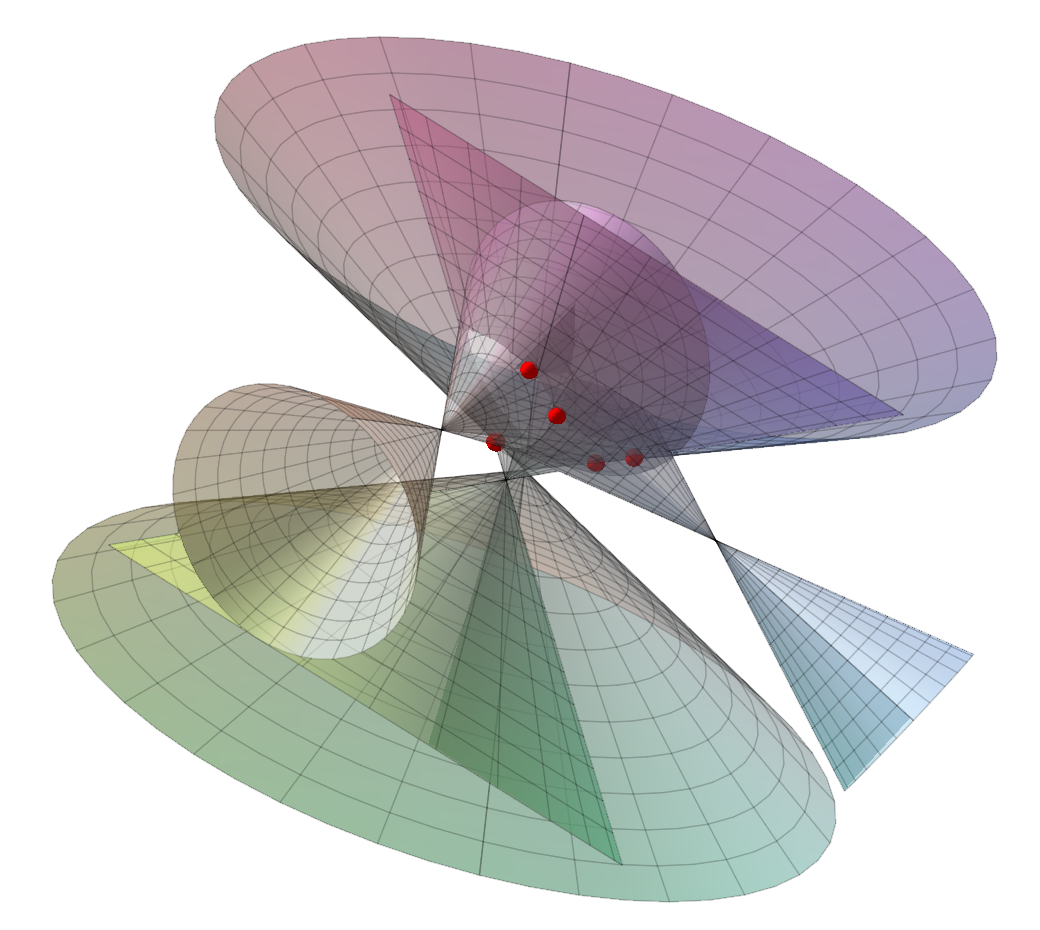}
   \includegraphics[width=40mm]{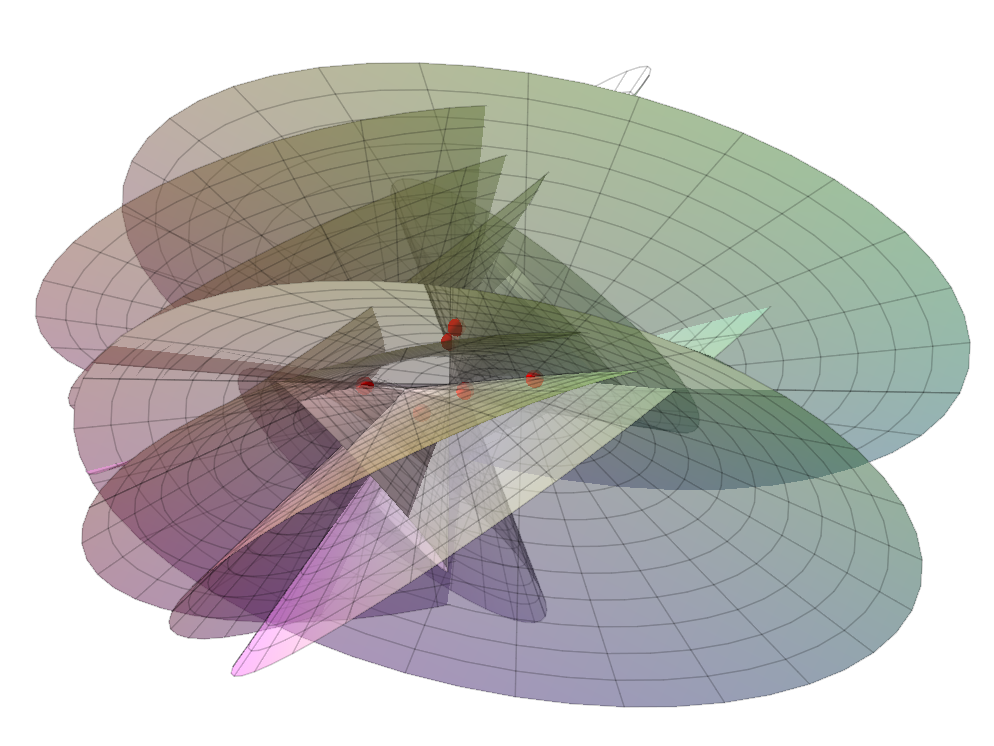}
   \includegraphics[width=40mm]{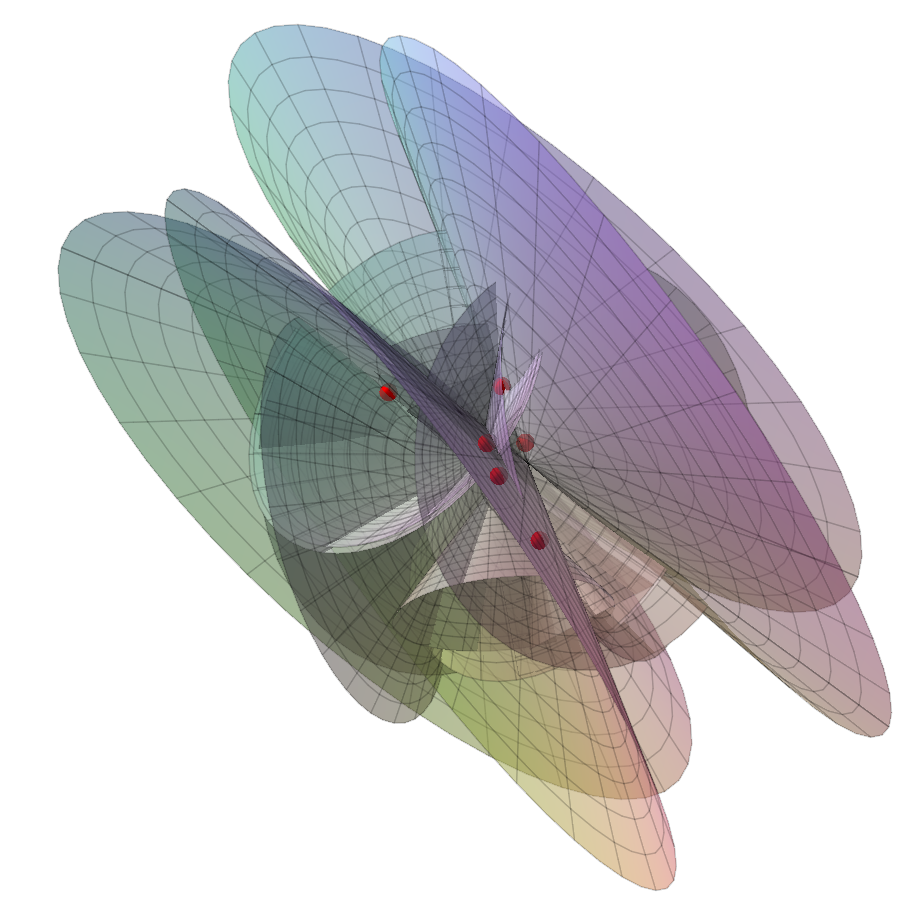}
   \includegraphics[width=40mm]{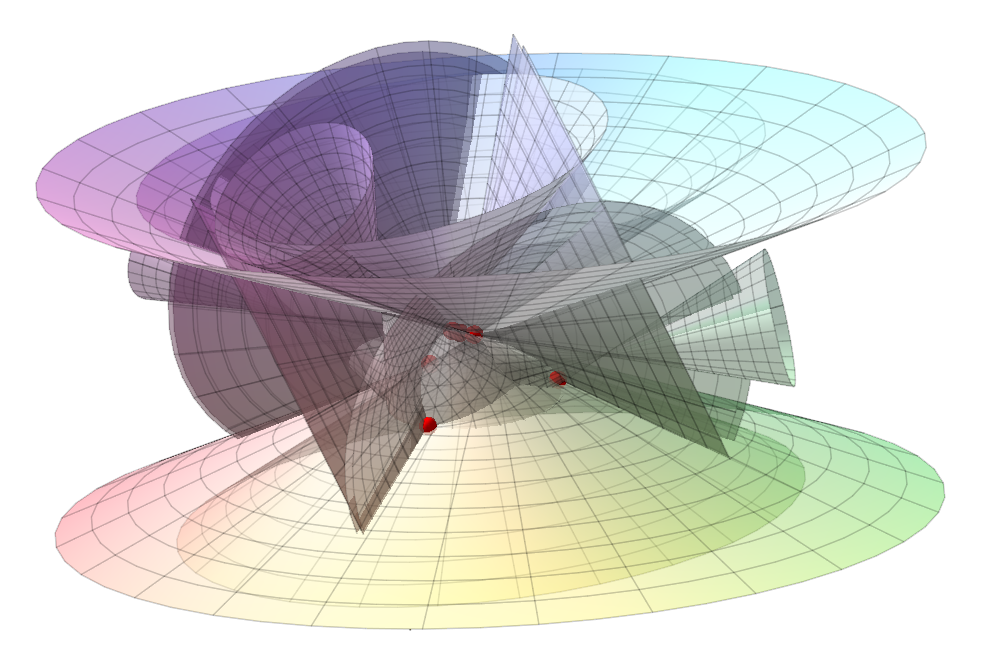}
   \includegraphics[width=40mm]{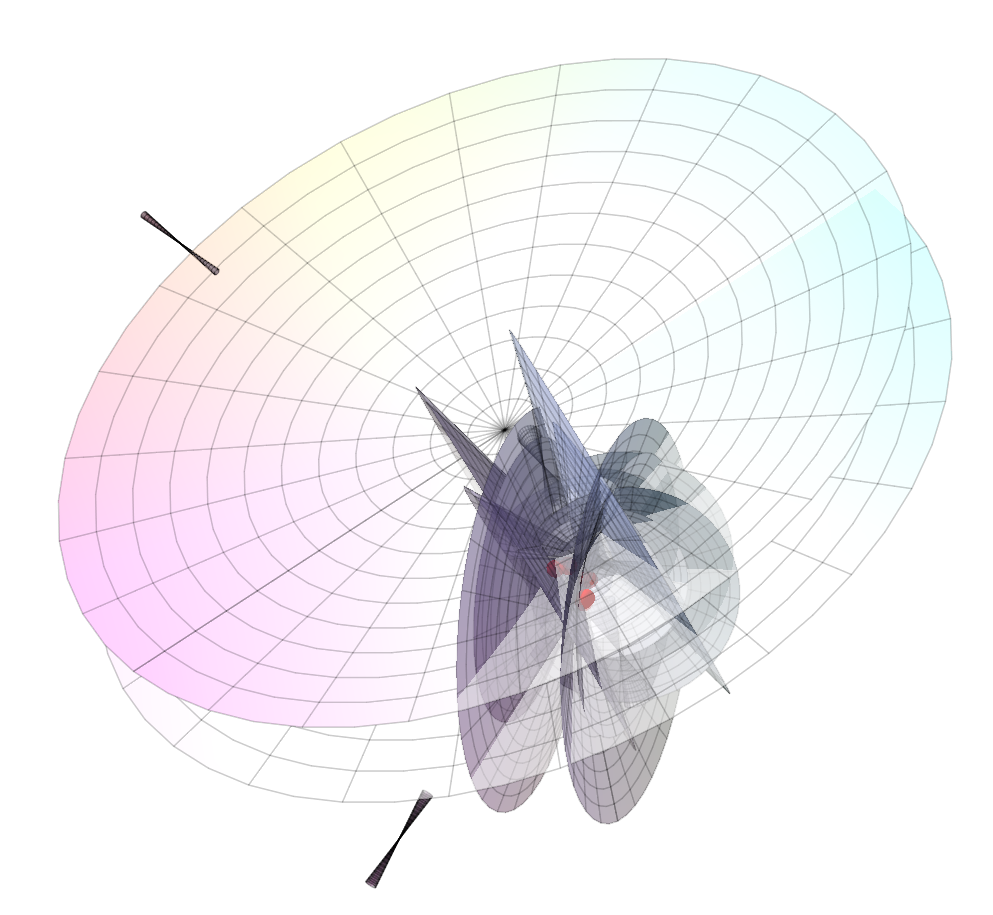}
   \caption{Configurations of points $\Pc$ with 2, 4, 6, 8, 10 and 12 cones passing through $\Pc$.}
   \label{fig:cones}
\end{figure}  

\paragraph{Extraction algorithm}

Summarizing the above analysis, \com{solving for the} cones through $\Pc$ can be done as follows. 

\begin{enumerate}
	\item Perform a linear change of coordinates so that the five points are of the form \eqref{eq:hyp-p1p2p3p4p5p6}. 
	\item If $x_2$, $y_3$ or $z_4$ is equal to zero then we are in a particular configuration (see Section \ref{subsec:cone6P-particular}) \com{so this} algorithm stops here.  
	\item Build the system of equations \eqref{eq:cone6Psys} which is given in closed form in terms of the coordinates of the input points.
	\item Solve the system $E_0(u,v)=E_1(u,v)=0$ in the two variables $u,v$ by means of eigenvectors and eigenvalues computations, as explained in Section \ref{subsec:cyl5P} (see \cite{BKM05} for details). \com{This computation} returns a list of an even number between 0 and 12 of \com{real} solutions (up to a given precision).  
	\item For each above solution $(u_i,v_i)$, compute the corresponding value of $w_i$ \com{using} the equation $E_2(u_i,v_i,w_i)=0$. \com{Then pullback} the solution $(u_i,v_i,w_i)$ under $\sigma$ to get the solutions $(l_i=u_in_i,m_i=v_in_i,n_i=\sqrt{w_i})$ of the polynomial system $H_0=H_5=H_6=0$. 
	\item For each solution $(l_i,m_i,n_i)$, compute the corresponding $(a_i,b_i,c_i)$ (as explained in the proof of Theorem \ref{thm:12cones}) and return the corresponding cone through $\Pc$.
\end{enumerate} 

This algorithm has been implemented with the {\tt Maple} software. We observed that \com{computing} the cones through a random set of points takes in average 80ms \com{(including all the steps in the above algorithm)} and is almost constant, \com{i.e.~independent} of the point set. In Table \ref{tab:cone6P}, we provide the \com{distribution} of the number of cones through a set of \com{points} $\Pc$ for a random sample of a thousand point sets $\Pc$.

\begin{table}
\begin{center}
\begin{tabular}{c|c|c|c|c|c|}
 Number of cones & 0 & 2 & 4 & 6  \\ 
	 \hline
Proportion (\%)	 &  1 & 10,5  & 28,7 & 36,3   \\ \hline \hline

 Number of cones  & 8 & 10 & 12 & \\ 
	 \hline
Proportion (\%)	  & 18,7 & 3,9 & 0,9 & \\

\hline	
\end{tabular}
\caption{Proportion of the number of cones found through a thousand random point sets.}\label{tab:cone6P}
\end{center}
\end{table}

\paragraph{Some particular configurations}\label{subsec:cone6P-particular}

If every subset of four \lb{po\-ints} in $\Pc$ is coplanar then all the six points in $\Pc$ are necessarily coplanar. In such a configuration, there exists a cone through $\Pc$ if and only if $\Pc$ can be interpolated by a conic section (parabola, hyperbola or ellipse). If this is the case, then there are infinitely many cones through $\Pc$. When a subset of five points in $\Pc$ are coplanar then we get interesting particular configurations.

\begin{thm}\label{them:4cones} \com{Given} a point set $\Pc$ where five points are coplanar but not six. \lb{There} exists a cone through $\Pc$ if these five points are located on \com{an ellipse or a hyperbola}. Moreover, if this is the \com{case there} are at most four cones through $\Pc$. 
	
	More specifically, if the five points are co-circular, then there are two cones through $\Pc$ and if \com{the five points} are on a parabola then there are at most 3 cones through $\Pc$. 
\end{thm}

\begin{proof}It is clear that the five coplanar points must be on a conic section if there exists a cone through $\Pc$. By a change of coordinate system, one can assume that this conic section \com{lies in the plane} $z=0$ and is in its canonical form. We examine below the three cases that correspond to an ellipse, \com{a hyperbola} and a parabola. 

\medskip	
	
\noindent \textit{Ellipse.} The equation of the conic section is of the form 
$$\alpha x^2 + \beta y^2 -1 =0, \ \ \alpha,\beta > 0$$
and it must coincide with the intersection of the implicit equation of a cone \eqref{eq:eqcone} after the substitution $z=0$. \com{Now} we proceed by \lb{identifying} the monomial coefficients of these two equations. the coefficient of the monomial $xy$ yields the condition $2lm=0$ so that $l=0$ or $m=0$. If $l=0$, then the coefficients of $y^2$ yields the condition $m^2=1-\beta/\alpha$ and if $m=0$ then the coefficients of $x^2$ yields the condition $l^2=1-\alpha/\beta$. Depending on the ratio $\alpha/\beta$, only one case leads to real solutions. From now on we assume that $\alpha/\beta\geq 1$, the other case can be treated similarly. So we have that 
$$ l=0, \ \ m^2=1-\frac{\beta}{\alpha}\geq 0.$$
The coefficient of the monomial $x$ yields the condition $a=0$ and the coefficient of $y$ gives the condition
\begin{equation}\label{eq:coneellipseb}
b=-\frac{cmn}{m^2-1}=\frac{\alpha}{\beta}cmn	
\end{equation}
(observe that $m^2-1\neq 0$). 
Finally, the coefficients of the monomial $1$ yields the condition
\begin{equation}\label{eq:ellipseq1}
(\alpha n^2 - \beta) c^2 - \frac{\beta}{\alpha}=0.	
\end{equation}
\com{Now} let $(x_0,y_0,z_0)$ be the coordinates of the sixth point. \com{Evaluating} \eqref{eq:eqcone} at this point together with the conditions $l=0, a=0$ and \eqref{eq:coneellipseb} yields the condition 
\begin{equation}\label{eq:ellipseq2}
  {c}^{2}\left({n}^{2} -\frac{\beta}{\alpha} \right) -2cz_0 \left( 
   {n}^{2} -\frac{\beta}{\alpha}\right) + \varphi_2(n)=0	
\end{equation}
(recall that $\alpha,\beta >0$ and in particular $m^2-1\neq 0$) where
\begin{multline*}
\varphi_2(n):= 
  \left( -{m}^{2}{z_0}^{2}+{z_0}^{2} \right) {n
  }^{2} + \left( -2\,{m}^{3}y_0z_0+2\,my_0z_0 \right) n \\ -{m}^{4}{y_0}^{2}+{m}^{2}{x_0}
  ^{2}+2\,{m}^{2}{y_0}^{2}+{m}^{2}{z_0}^{2}-{x_0}^{2}-{y_0}^{2}-{z_0}^{2}	
\end{multline*}
is a degree 2 polynomial in $n$ \com{that} is independent of $c$. The cones through $\Pc$ are in correspondence with the solutions of the equations \eqref{eq:ellipseq1} and \eqref{eq:ellipseq2} in the variables $c$ and $n$. Using \eqref{eq:ellipseq1}, Equation \eqref{eq:ellipseq2} becomes
$$ 2cz_0\left(\alpha{n}^{2} -\beta \right)=\alpha \varphi_2(n)+\frac{\beta}{\alpha}$$
and hence we get
$$ \left( \alpha \varphi_2(n)+\frac{\beta}{\alpha}   \right)^2=4c^2z_0^2(\alpha{n}^{2} -\beta)^2= 
4z_0^2\frac{\beta}{\alpha}(\alpha{n}^{2} -\beta)
$$
which is an equation of degree 4 in $n$. Therefore, once $m$ is chosen from the condition $m^2=1-\beta/\alpha$, we obtain at most 4 values of $n$ and all the \com{other} parameters are uniquely determined. We conclude that we have at most four cones through $\Pc$, as claimed. 

If the conic section is a circle, i.e.~$\alpha=\beta$, we have $l=m=0$. In this case, $\varphi_2(n)=z_0^2(n^2-1)-(x_0^2+y_0^2)$ and hence the degree 4 equation in $n$ is actually an equation in $n^2$. Therefore, the 4 solutions \com{come in} pairs of opposite solutions and we deduce that we have two cones through $\Pc$.

\medskip	
	
\noindent \textit{Hyperbola.} The equation of the conic section is of the form   
	$$\alpha x^2 - \beta y^2 -1 =0, \ \ \alpha,\beta > 0$$
and we proceed similarly to the case of the ellipse. The identification of $xy$ implies that $2lm=0$ so we have \com{two} cases to consider, namely $l=0$ and $m=0$. 

If $l=0$ then the coefficient of $y^2$ yields $m^2=1+\beta/\alpha$. The coefficient of $x$ shows that $a=0$ and the coefficient of $y$ gives
$$b=-\frac{\alpha}{\beta}cmn.$$ 
\com{Then} the constant term yields the equality
$$b^2+c^2-(bm+cn)^2=-\frac{1}{\alpha}$$
\com{that} becomes, after \com{substituting} $b$, 
$$c^2\left(1+\frac{\alpha}{\beta}n^2\right)=-\frac{1}{\alpha}.$$
Therefore, there are no real solutions in this case. 

\com{Now} if $m=0$, then $l^2=1+\alpha/\beta$. \com{The} coefficient of $y$ shows that $b=0$ and the coefficient of $x$ gives
$$a=-\frac{ncl}{l^2-1}=-\frac{\beta}{\alpha}ncl.$$   	
\com{Then} the constant coefficient yields the equation
$$\frac{c^2( l^2+n^2-1)}{l^2-1}=\frac{1}{\beta}.
$$ 
From here, we use the sixth point as in the case of the ellipse and we get again 4 possible real cones through $\Pc$.

\medskip	
	
\noindent \textit{Parabola.} The equation of the conic section is of the form 
$$x^2-\alpha y=0, \ \ \alpha> 0.$$	
\com{Again} we inspect the coefficients of \eqref{eq:eqcone} after \com{substituting} $z=0$. The coefficient of $x^2$ gives $l=0$ and the coefficient of $y^2$ gives $m^2=1$. The coefficient of $x$ gives $a=0$ and the coefficient of $y$ gives $2cmn=-\alpha$. Finally, the constant coefficient yields 
$$b=-\frac{c^2(1-n^2)}{\alpha}.$$
\com{Now} choosing $m=1$ and using the sixth point of coordinates $(x_0,y_0,z_0)$, \com{Equation}  \eqref{eq:eqcone} gives an equation of degree 3 in $n$, namely
\begin{multline*}
-2{z_0}^{2}{n}^{3}+ \left( -{\alpha} z_0-4 y_0 z_0 \right) {n}^{2} \\ + \left( 
-2 y_0{\alpha}+2{x_0}^{2}+2{z_0}^{2} \right) n+{\alpha}z_0=0	
\end{multline*}	
and from here $c$ can be uniquely determined.
\end{proof}

\section{Conclusion}

\com{We} have presented several methods in order to extract efficiently cylinders and cones from minimal point sets. We have also provided a detailed analysis of these interpolation problems and we have given optimal bounds on the number of solutions. Our approach relies on closed algebraic formulas that have been computed and experimented with the help of a computer algebra system. In \com{the} near future, we plan to \com{incorporate} these methods into an efficient C++ library \com{in order to experiment our new methods in the framework of a RANSAC-based extraction algorithm}. Another future research direction will be \com{extracting} tori from a minimal point set, tori being described by means of seven parameters. 

\section*{Ackowledgments}

The authors are grateful to Pierre Alliez for several interesting discussions about RANSAC-type algorithms and the extraction of geometric primitives in 3D point sets. \com{We also thank the anonymous reviewers for their helpful comments}.


\end{document}